\documentclass{article} %
\newcommand{\sv}[1]{}
\newcommand{\aaai}[1]{}
 \newcommand{\ijcai}[1]{}
 \newcommand{\lv}[1]{#1}
\usepackage{etoolbox}
 \newcommand{\toappendix}[1]{#1}

\aaai{
\usepackage[]{aaai2026}  %
\usepackage{times}  %
\usepackage{helvet}  %
\usepackage{courier}  %
\usepackage[hyphens]{url}  %
\usepackage{graphicx} %
\urlstyle{rm} %
\usepackage{natbib}  %
\usepackage{caption} %
\frenchspacing  %
\setlength{\pdfpagewidth}{8.5in} %
\setlength{\pdfpageheight}{11in} %
\pdfinfo{/TemplateVersion (2026.1)}
}
\ijcai{
\pdfpagewidth=8.5in
\pdfpageheight=11in
\usepackage{ijcai25}
\usepackage{times}
\usepackage{soul}
\usepackage{url}
\usepackage[hidelinks]{hyperref}
\usepackage[utf8]{inputenc}
\usepackage[small]{caption}
\usepackage{graphicx}
\usepackage{booktabs}
\usepackage{algorithm}
\usepackage{algorithmic}
\usepackage[switch]{lineno}
\usepackage{natbib}  %

\linenumbers

\urlstyle{same}

\author{}
}

\lv{
\usepackage{authblk}
\usepackage[T1]{fontenc}
\usepackage{fullpage}
\usepackage{hyperref}
\usepackage{graphicx} %
}

\usepackage{amsmath}
\usepackage{amsthm}
\usepackage{amssymb}
\usepackage{mathtools}
\usepackage{complexity}
\usepackage{xcolor}
\usepackage{tabularx}
\usepackage[capitalize,noabbrev]{cleveref}
\usepackage{tikz}
\usetikzlibrary{shapes.geometric}
\usetikzlibrary{positioning,calc,arrows.meta}

\tikzset{venn/.style={ellipse, draw, thick,
    minimum height=3cm, minimum width=6cm, rotate=#1
}}

\tikzset{ellip/.style={ellipse, draw, thick,
    minimum height=3cm, minimum width=6cm, rotate=#1
}}

\usepackage{comment}

\usepackage{thm-restate}

\newtheorem{definition}{Definition}[section]
\newtheorem{theorem}[definition]{Theorem}
\newtheorem{lemma}[definition]{Lemma}

\newtheorem{claim}[definition]{Claim}

\newtheorem{observation}[definition]{Observation}
\newtheorem{example}[definition]{Example}

\newtheorem{corollary}[definition]{Corollary}

\renewcommand{\epsilon}{\varepsilon}

\title{Finding Diverse Solutions Parameterized by Cliquewidth\thanks{The short conference version of this paper will appear in the proceedings of the 40th Annual AAAI Conference on Artificial Intelligence (AAAI 2026)~\cite{ourAAAI}.\\
  Both authors of this research were supported by the Polish National Science Centre grant number 2021/43/D/ST6/03312 (SONATA-17).
}}
  \author[1]{Karolina Drabik}%
  \author[1]{Tomáš Masařík}

  \affil[1]{Institute of Informatics, Faculty of Mathematics, Informatics and Mechanics, University~of~Warsaw, Warszawa, Poland 

 \texttt{kd417818@students.mimuw.edu.pl}\hfil \texttt{masarik@mimuw.edu.pl}}
\date{}

\begin{document}

\maketitle 
\begin{abstract}

Finding a few solutions for a given problem that are diverse, as opposed to finding a single best solution to solve the problem, has recently become a notable topic in theoretical computer science.
It is important to mention that finding a single solution is computationally hard in most interesting cases.
To overcome this issue, we analyze the problems under a finer scope of parameterized complexity, which often allows for tractable algorithms.
That, in turn, enables us to find a diverse set of solutions.

Recently, Baste, Fellows, Jaffke, Masařík, Oliveira, Philip, and Rosamond showed that under a standard structural parameterization by treewidth, one can find a set of diverse solutions for many problems with only a very small additional cost [Artificial Intelligence 2022].
In this paper, we investigate a much stronger graph parameter, the cliquewidth, which can additionally describe some dense graph classes.
Broadly speaking, it describes graphs that can be recursively constructed by a few operations defined on graphs whose vertices are divided into a bounded number of groups, while each such group behaves uniformly with respect to any operation.

We show that for any vertex problem, if we are given a dynamic program solving that problem on cliquewidth decomposition, we can modify it to produce a few solutions that are as diverse as possible with as little overhead as in the above-mentioned treewidth paper.
As a consequence, we prove that a diverse version of any MSO$_1$ expressible problem can be solved in linear FPT time parameterized by the cliquewidth, the number of sought solutions, and the number of quantifiers in the formula, which was a natural missing piece in the complexity landscape of structural graph parameters and logic for the diverse problems.
We prove our results, allowing for a more general natural collection of diversity functions compared to only two mostly studied diversity functions previously.
That might be of independent interest as a larger pool of different diversity functions can highlight various aspects of different solutions to a problem.

 \bigskip{}
  \begin{center}
\end{center}
\end{abstract}

\aaai{
\begin{links}
  \link{Extended version}{https://arxiv.org/abs/2405.20931}
\end{links}
}

\section{Introduction}
In a standard graph problem, one typically asks for a solution of the smallest size.
The first difficulty that typically needs to be overcome is that a large proportion of important problems is hard to solve (for some notion of hardness). 
There is a vast branch of research dedicated to overcoming such a difficulty, like parameterized complexity, approximation algorithms, and heuristics, among others.
Second, maybe a more pressing issue is that even once we are given a single optimal solution for the problem, we might not be able to learn as much about our original tasks as we had hoped.
Why might it be so?
Mathematically formulated problems are typically just an abstraction of real-world problems, where there might be other important factors that were perhaps just too complex to formalize or even not known in the first place.
Such constraints might be visible only once an actual solution is presented.
While the first difficulty stands, the second difficulty might be overcome or at least diminished by the following idea.
Instead of searching for just one solution, we will search for various solutions that are as different as possible (for some notion of what different means, to be specified later).
Moreover, the idea of looking for a diverse set of solutions may even further allow us to relax the original mathematical formulation of a problem, which in turn may allow for more efficient algorithms combating the first-mentioned difficulty.
There are many examples of advantages producing numerous different solutions in the literature; see the famous architect problem~\cite{architect} as one of the earliest examples.
Since then, there have been countless applications of this principle in various branches of computer science: constraint programming~\cite{
  HOW07,
  HHOW05,
  PT15,
IGST20}, mixed integer optimization~\cite{
  AMT24,%
  DW09, 
  GLW00},
  evolutionary computation~\cite{DGNN22,GLPS18,WO03}, social choice~\cite{BN21, %
  AFLOW21}, %
 and discrete geometry~\cite{KK22} among many others.

It was not until 2019 that this concept was brought to the attention from the perspective of theoretical graph algorithms and parameterized complexity toolbox~\cite{BFJMOPR22}.
They showed that a dynamic program (DP) on a tree decomposition can be converted to a DP finding diverse solutions.
Since then, this concept under the parameterized complexity examination started to play a more prominent role, and many follow-up papers have been published extending our knowledge into various directions: diverse hitting sets~\cite{%
BJMPR2019%
}, diverse polynomial-time solvable problems (spanning trees, paths)
\cite{
HKKO21}, %
diverse minimum $st$-cuts~\cite{
BLS23, BLS25, FlowFramework%
}, diversity in matroids~\cite{
FGPPS2023%
} approximation algorithms framework for diverse solutions~\cite{
HMKK23}, %
diverse synchronization of words~\cite{
AFOW23}, %
Submodular function approximations and experiments~\cite{
DGNN23}, %
diverse SAT~\cite{MMR24},
and diverse pairs of matchings~\cite{
FGJPS24}. %

Now, let us define diverse problems formally.
We start by explaining the two basic diversity measures that were within the central focus of the theoretical research on diverse problems.
We define the \emph{Hamming distance} between two sets $S$ and $S'$ as the size of their symmetric difference.
That is
\[HamDist(S, S') \coloneqq |(S \setminus S' ) \cup (S' \setminus S)|.\]
We use it to define diversity measures of an (ordered) collection of sets.
We define \emph{Sum diversity} as
\begin{equation*}
  DivSum(S_1, \ldots, S_r) \coloneqq \sum_{1 \leq i < j \leq r} HamDist(S_i, S_j).
\end{equation*}
Then, we define \emph{Min diversity} as
\begin{equation*}
  DivMin(S_1, \ldots, S_r) \coloneqq \min_{1 \leq i < j \leq r} HamDist(S_i, S_j).
\end{equation*}
Having defined the basic measures, we can state the problem formally.
In the following, we use $Div$ to denote an arbitrary diversity measure that takes $r$ subsets of vertices and outputs an integer, for example, $DivSum$ or $DivMin$ functions defined above.
Note that in this work, we are concerned only with vertex problems, but this template can be easily modified to take into account any list of subsets.
Formally, a \textit{vertex problem} is a set of pairs $(G, S)$, where $G$ is a graph and $S$ is a subset of its vertices that satisfies some property.
\lv{For example, we can define k-\textsc{Vertex Cover} as the set of pairs $(G, S)$ such that $|S| \leq k$ and $S$ is a vertex cover of $G$.}

\begin{definition}[Diverse Problem] 
  Let $\mathcal{P}_1, \ldots, \mathcal{P}_r$ be vertex problems, %
  $Div$ be a diversity measure, $d \in \mathbb{N}$. Then
  \begin{align*}
    Div^d(\mathcal{P}_1, \ldots, \mathcal{P}_r) \coloneqq \{(G, S_1, \ldots, S_r)|\sv{\\  } (G, S_i) \in \mathcal{P}_i, Div(S_1, \ldots, S_r) \geq d\}.
  \end{align*}
\end{definition}
\lv{%
In our paper, we focus on FPT algorithms for diverse vertex problems, parameterized by the \emph{cliquewidth} of the graph.

  \medskip
A graph $G$ has cliquewidth $\omega$ if there is a rooted tree that describes the construction of $G$ from single vertices via the following operations that use at most $\omega$ colors at any time during the construction.
The operations are: disjoint union of two graphs, recolor all vertices of one color to another one, and add all edges with endpoints of specified colors.
Inspect \Cref{sec:cw}} for a formal exposition.

The cliquewidth has played an influential role since its definition in 1993~\cite{cw-def}.
This has only increased after a seminal result of Courcelle, Makowsky, and Rotics proving that any MSO$_1$ property can be checked in linear time on graphs of bounded cliquewidth~\cite{CMR}.
This result essentially means that many difficult vertex problems on graphs are solvable in linear time, provided that the cliquewidth of that graph is small.
The most prominent tool to solve problems on graphs of bounded cliquewidth is dynamic programming without any doubt. 
Besides the model-checking theorem, there have been countless examples of using dynamic programming to solve some particular problems, either because those problems are not expressible in MSO$_1$ or because they provide a much faster algorithm. 
For example, feedback vertex set~\cite{BSTV13}, triangle detection \& girth~\cite{CDP19}, Alliances~\cite{KO17}, Hamiltonian cycle~\cite{BKK19}, Steiner tree~\cite{BK24Stree}, connected odd cycle transversal~\cite{BK24}, and other connectivity problems~\cite{BK19,HK23}.

We would like to point out that for our purposes, there is very little difference between the case when the cliquewidth decomposition is given on the input and the case when it first needs to be computed.
Indeed, despite being NP-hard to be computed optimally \cite{FRRS09}, there is a $(2^{2\omega +1}-1)$-approximation algorithm running in quadratic time, where $\omega$ is the cliquewidth of the graph on the input~\cite{FK22}.
Therefore, from now on, we will simply assume that the decomposition is given on the input, so we do not need to argue this in many places.
Moreover, it allows us to state our results independently of potential improvements of the above-mentioned approximation algorithm.

\subsection{Our Results}
Our main contribution is the following.
Given a cliquewidth decomposition of a graph $G$ and (almost any) dynamic program solving a vertex problem on $G$, we are able to translate it to an algorithm solving a diverse version of the problem with only a very small expense on the running time.
In fact, the running time of the diverse algorithm will be essentially the $r$-th power of the running time required to compute a single solution when we aim to find $r$ diverse solutions. 

More formally, we define what it means for dynamic programming to be monotone; see \Cref{def:monotone} for a formal definition.
Despite its inherent technical nature, the idea behind this definition is very natural.
Let us explain the intuition here. %
Take a dynamic programming table entry $\epsilon$ within some inner node of the DP.
Such $\epsilon$ could originate from many different possible combinations of other entries evaluated earlier (in subtrees of the given node) in the DP.
We call it the \emph{derivation tree} of $\epsilon$.
  To each such derivation tree, we assign a set of vertices, called \emph{partial solution}, based only on the entries in the leaves of the decomposition that combine to $\epsilon$.
Moreover, for the leaf node, we only allow for two entries: one where the vertex corresponding to a leaf node is in the partial solution and the other when it is not.
We call the described property monotonicity, as the set of vertices assigned to the derivation tree is the union of the sets assigned to its subtrees.

With the intuition on what monotone DPs look like, we formulate our main theorem.
In the following, we use $Div$ to denote an arbitrary diversity measure from a collection of Venn measures that will be defined later; see \Cref{def:venn} for a formal definition. 
At this point, we keep in mind that $DivSum$ is one such measure, but $DivMin$ is not. 
\begin{theorem}[Diverse problems on cw] \label{mainThm}
  Let $\mathcal{P}_1, \ldots, \mathcal{P}_r$ be vertex problems and given monotone dynamic programs solving them with the slowest running on a single node of the decomposition in time $t_{\mathcal{D}}$ for graph $G$ with cliquewidth decomposition~$\mathcal{D}$. 
  Let $Div$ be a Venn diversity function.
  The problem $Div^d(\mathcal{P}_1, \ldots, \mathcal{P}_r)$ on a graph $G$ with given cliquewidth decomposition $\mathcal{D}$ can be solved in time 
  $\mathcal{O}(|V(\mathcal{D})|t_{\mathcal{D}}^r)$.
\end{theorem}

It is important to note that the running time of \Cref{mainThm} does not depend on the actual value of the $Div$ function.
Our proof also provides the same statement for the $DivMin$ function, but with a worse running time dependent on the target value $d$ of the $DivMin$ function.
\lv{In particular, the running time will be $\mathcal{O}(|V(\mathcal{D})|t_{\mathcal{D}}^r d^{r^2})$ in that case.}
\begin{theorem}[Min Diverse problems on cw] \label{minmainThm}
  Let $\mathcal{P}_1, \ldots, \mathcal{P}_r$ be vertex problems and given monotone dynamic programs solving them with the slowest running on a single node of the decomposition in time $t_{\mathcal{D}}$ for graph $G$ with cliquewidth decomposition~$\mathcal{D}$. 
  The problem $DivMin^d(\mathcal{P}_1, \ldots, \mathcal{P}_r)$ on a graph $G$ with given cliquewidth decomposition $\mathcal{D}$ can be solved in time 
  $\mathcal{O}(|V(\mathcal{D})|t_{\mathcal{D}}^r d^{r^2})$.
\end{theorem}

We would like to compare the running time of \cref{mainThm} with an analogous theorem for the treewidth \cite[Theorem 10]{BFJMOPR22}.
We present it with a minor improvement, where we omit the factor $d^a$ (typically ${a=2}$) stated in the referenced running time as it is not necessary once one aims to find a single diverse solution of the best possible value of $DivSum$.
We discuss the change at the end of \cref{sec:mainThm}.

\begin{theorem}[Reformulated {\cite[Theorem 10]{BFJMOPR22}} with a slight improvement] \label{twThm}
  Let $\mathcal{P}_1,  \ldots, \mathcal{P}_r$ be vertex problems and given dynamic program core (as defined in \cite[Definition 6]{BFJMOPR22}) solving them with the slowest running on a single node of the decomposition in time $t_{\mathcal{T}}$ for graph $G$ with tree decomposition~$\mathcal{T}$.
  The problem $DivSum^d(\mathcal{P}_1,  \ldots , \mathcal{P}_r)$ on a graph $G$ with given treewidth decomposition $\mathcal{T}$ can be solved in time
  $\mathcal{O}(|V(\mathcal{T})|t_{\mathcal{T}}^r)$.
\end{theorem}

\medskip

We claim that a vast majority of dynamic programs solving vertex problems using the cliquewidth decomposition are either monotone or could be easily transformed to be monotone.
To the best of our knowledge, we have not found any examples of non-monotone DPs in the literature.
To support the intuition, we designed a non-monotone DP for the vertex cover problem; see \Cref{sec:nonMonotone}.
  However, this algorithm can be easily turned into a monotone one, which we present in \Cref{sec:monotoneVC}.
In order to support our claim, we showed that using our framework, one can solve the diverse version of any vertex problem that can be expressed in MSO$_1$ logic, which is the fragment of second-order logic where the second-order quantification is limited to quantifying over sets of vertices.
The MSO$_1$ captures many well-studied problems, including: Vertex Cover, Independent Set, Dominating Set, and Feedback Vertex Set.
We show that any problem $\mathcal{P}$ for which there exists an MSO$_1$ formula $\Phi(S)$ such that $G \vDash \Phi(S)$ if and only if $(G, S) \in \mathcal{P}$ can be solved by a monotone dynamic program to which we apply \Cref{mainThm}.
As the original result~\cite{CMR} gives the DP only implicitly, we decided to design a new DP algorithm and show that it is indeed monotone.
We based our approach on lecture notes~\cite[Section 5.4]{Fiala}\footnote{This approach originates in an unpublished work of Hliněný, Král’, and Obdržálek.}, which show how to use evaluation trees to obtain a similar result for treewidth decomposition.
\begin{corollary}\label{cor:mainCor}
  Diverse MSO$_1$ problems can be solved in linear FPT time parameterized by the cliquewidth, the number of solutions, and the number of quantifiers in the formula.
\end{corollary}

\Cref{cor:mainCor} is interesting on its own as it provides a new result for graph vertex problems expressible in certain logic.
We can compare it with the original treewidth result~\cite{BFJMOPR22}, which encompasses problems that can be defined in MSO$_2$ logic and then provides the diverse variant in the running time comparable to \Cref{cor:mainCor}. 
That is a stronger logic, but a much weaker graph class.
Note that model checking of MSO$_2$ problems on cliques is not even in XP unless $\text{E}=\text{NE}$ \cite{CMR}.
Later, Hanaka, Kobayashi, Kurita, and Otachi~\cite[Theorem 13]{HKKO21} proved that problems that can be defined in FO logic on a nowhere dense graph class can have their diverse variant solved in FPT time.
However, this time the result relies heavily on the model-checking machinery \lv{of Grohe, Kreutzer, and Siebertz~}\cite{GKS17} and the claimed running time is much worse than linear.
Again, this result is incomparable with \Cref{cor:mainCor} as it acts on a larger class of graphs, but only on problems defined in a weaker logic. 
In 2023, Bergougnoux, Dreier, Jaffke~\cite{BDJ23} worked on an A\&C DN logic (distance neighborhood logic with acyclicity and connectivity constraints).
Informally speaking, DN logic allows for size measurement of neighborhood terms as well as inclusion-wise comparison of neighborhoods within the \emph{existential} fragment of MSO$_1$ logic, while A\&C further add two more constraints which allow us to express acyclicity and connectivity of the neighborhood term.
See \cite[Section 3]{BDJ23} for a formal definition.
They showed that if a vertex problem can be expressed in A\&C DN logic, then one can express $DivMin$ or $DivSum$ variant of the problem in that logic as well~\cite[Observation~1.1]{BDJ23}.
Then they present a cubic FPT algorithm solving A\&C DN encoded problems on graphs of bounded cliquewidth parameterized by the size of the formula and the cliquewidth.
Additionally, they present an XP algorithm for the same problem on graphs of bounded mimwidth\footnote{Informally, the mimwidth of a graph minimizes the size of each induced matching appearing in some particular cuts of the graph which are specified by a recursive decomposition of its vertex set; see \cite{Vat} for a formal definition.}
which is a much more general structural graph parameter.
Indeed, as this covers only an existential fragment of MSO$_1$, it is incomparable with \Cref{cor:mainCor}.
Let us stress that all the related work is stated only for $DivSum$ or possibly for $DivMin$ measures.

\subsubsection{New Venn Diversity Measures.}\label{sec:intro-measures}
In our work, we propose a whole collection of diversity measures.
We believe they will be of independent interest.
Despite being very natural, to the best of our knowledge, they have not been studied in this context before.
In the previous theoretical works, research was only focused on the two measures already defined: $DivMin$ and $DivSum$.
However, despite them being the first natural examples, there is nothing so special about them.
One can find examples where those measures do not behave exactly as expected.
In the case of the $DivSum$ measure, a profound example is taking many copies of two disjoint sets, which results in a relatively high diversity value. Still, intuitively, such a tuple of sets is not diverse.
In particular, it does not display a potentially rich structure of various different solutions as we might have hoped.
This was already observed and discussed in \cite[Section 5, Figure 3]{BJMPR2019}; see also an example illustrated later in~\cref{fig:ds}.
Such problems are not limited to $DivSum$.
See \Cref{ex:DivMin} (\cref{fig:DivMin}), where the optimal $DivMin$ value is $2$ when seeking six solutions, while there is another tuple of six solutions with the same $DivMin$ value, which seems intuitively much more diverse.
Indeed, a straightforward generalization of \cref{ex:DivMin} gives that $DivSum$ of the first $r$-tuple of solutions equals $2\binom{r}{2}$ and the second $r$-tuple has $DivSum$ arbitrarily high (depending only on $s$, as $r$ is fixed) while they both have the same $DivMin=2$.
Similarly to the examples above, we believe that for each potential diversity measure, some problematic instances can be created so that when the specific measure is used, the most diverse collection does not reflect at all how diverse the whole space of solutions might be.

\begin{figure}[t]
\centering
\begin{tikzpicture}[
  scale=1.2,
  font=\small,
  Bnode/.style={circle,draw,inner sep=1.2pt},
  Anode/.style={rectangle,draw,rounded corners=2pt,minimum width=12mm,minimum height=6mm},
  allEdge/.style={draw=black!25},
  pick/.style={ultra thick},
  >=Latex
]

\begin{scope}[shift={(0,0)}]
 \node[font=\bfseries] at (2.2,5.25) %
 {First solution};
   \node[] at (2.2,4.75)%
 {$\{A_1\cup\{b_1\},\ldots,A_1\cup\{b_6\}\}$};
   \node[] at (2.2,4.25)%
{$DivMin=2,\quad DivSum=30$};

\foreach \i [count=\k from 0] in {1,...,6}{
  \node[Bnode] (aB\i) at (0,3.6-0.55*\k) {$b_{\i}$};
}
\node at (0,0.49) {$\vdots$};
\node[Bnode] (aBs) at (0,-0.09) {$b_s$};

\foreach \i [count=\k from 0] in {1,...,5}{
  \node[Anode] (aA\i) at (4,3.6-0.7*\k) {$A_{\i}$};
}

\node[align=center] at (0,-0.75) {$B$\\$|B|=s$};
\node[align=center] at (4,-0.75) {$A_1,\dots,A_5$\\$|A_i|=s$};

\foreach \bi in {1,...,6}{
  \foreach \ai in {1,...,5}{
    \draw[allEdge] (aB\bi) -- (aA\ai);
  }
}
\foreach \ai in {1,...,5}{
  \draw[allEdge,densely dotted] (aBs) -- (aA\ai);
}

\foreach \bi in {1,...,6}{
  \draw[pick] (aB\bi) -- (aA1);
}

\end{scope}

\begin{scope}[shift={(7.3,0)}]
 \node[font=\bfseries] at (2.2,5.25) %
 {Second solution};
   \node[] at (2.2,4.75)%
 {$\{A_i\cup\{b_i\}\mid i\in[5]\}\cup\{A_1\cup\{b_6\}\}$};
   \node[] at (2.2,4.25)%
{$DivMin=2,\quad DivSum=28s+30$};

\foreach \i [count=\k from 0] in {1,...,6}{
  \node[Bnode] (bB\i) at (0,3.6-0.55*\k) {$b_{\i}$};
}
\node at (0,0.49) {$\vdots$};
\node[Bnode] (bBs) at (0,-0.09) {$b_s$};

\foreach \i [count=\k from 0] in {1,...,5}{
  \node[Anode] (bA\i) at (4,3.6-0.7*\k) {$A_{\i}$};
}

\node[align=center] at (0,-0.75) {$B$\\$|B|=s$};
\node[align=center] at (4,-0.75) {$A_1,\dots,A_5$\\$|A_i|=s$};

\foreach \bi in {1,...,6}{
  \foreach \ai in {1,...,5}{
    \draw[allEdge] (bB\bi) -- (bA\ai);
  }
}
\foreach \ai in {1,...,5}{
  \draw[allEdge,densely dotted] (bBs) -- (bA\ai);
}

\draw[pick] (bB1) -- (bA1); %
\draw[pick] (bB2) -- (bA2); %
\draw[pick] (bB3) -- (bA3); %
\draw[pick] (bB4) -- (bA4); %
\draw[pick] (bB5) -- (bA5); %
\draw[pick] (bB6) -- (bA1); %

\end{scope}

\end{tikzpicture}
\caption{Illustration of \cref{ex:DivMin} drawn as a bipartite incidence graph.
Each feasible solution corresponds to an edge $(b_j,A_i)$ representing the set $A_i\cup\{b_j\}$.
} \label{fig:DivMin}
\end{figure}

\begin{example}\label{ex:DivMin}
   We fix the sought number of solutions $r=6$.
   We partition the vertices into the following sets $B,A_1,A_2,A_3,A_4,A_5$, each of the same size $s\ge 6$.
    Consider a problem where only the following sets are possible solutions: 
    any single vertex in $B$ together with the whole set $A_i$, $i\in[5]$.
    Then a possible set of six solutions with the maximum value of $DivMin$ ($=2$) consists of sets $A_1,b_j$, where $j\in[6]$.
    Such a solution has $DivSum$ equals $30$.
    However, there is a set of six solutions with intuitively bigger diversity (and also $DivSum = 28s+30$) while its $DivMin$ remains at $2$.
    For example, take $A_i,b_i$ for $i\in[5]$ and $A_1,b_6$.
    See \cref{fig:DivMin} depicting this example.
  \end{example}

Therefore, we are convinced that our approach, which gives access to the whole collection of diversity measures, much better supports the nature of the diverse problems.
Indeed, it seems in line with the diversity paradigm, where one seeks a collection of solutions rather than a single one.
In the same spirit, we advocate trying different diversity measures and then comparing their results.
However, we can be more subtle as we have access to a large collection of measures.
We suggest that one might pick the measures dynamically. 
This means we can choose the measures not only based on the instance but also based on the diverse collection(s) we have obtained so far.
If we obtain one of the degenerate examples, %
or simply, whenever we do not like the solution given by one diversity function, we can adjust the diversity function to penalize such an outcome and rerun the algorithm.

We will shortly give an example of a measure that penalizes an unwanted outcome returned using the $DivSum$ measure, but first, we define Venn diversity measures formally.
We start by defining the membership vector, which is basically a characteristic vector of a vertex $v$ being an element of set $S_i$ in the list of sets $S_1, \ldots, S_r$.

\begin{definition}[Membership vector]
  For $v \in V$ and a list $S_1, \ldots, S_r$ of subsets of $V$ a \emph{membership vector} of $v$ in $S_1, \ldots, S_r$
is a vector $m(S_1, \ldots, S_r, v) \in \{0,1\}^r$ such that for each $i \in [r]$,
  \begin{equation*}
    m(S_1, \ldots, S_r, v)[i] =
    \begin{cases}
      1 & \text{ if } v \in S_i, \\
      0 & \text{ otherwise.}
    \end{cases}
  \end{equation*}  
\end{definition}

Then we define the influence of a vertex provided with a function $f$ where the value of $f$  is determined by the membership vector value. 
\begin{definition}[$f$-influence]
  Let $f: \{0, 1\}^r \rightarrow \mathbb{N}$. For $v \in V$ and a list $S_1, \ldots, S_r$ of subsets of $V$ the \emph{$f$-influence} of $v$ on the list $S_1, \ldots, S_r$ is the number
  \begin{equation*}
      I_f(S_1, \ldots, S_r, v) \coloneqq f(m(S_1, \ldots, S_r, v)).
  \end{equation*}
\end{definition}
Observe that in this definition, the influence depends only on the presence of $v$ in each of $S_1, \ldots, S_r$.
Finally, we define the Venn function; see \cref{fig:venn}.
  
\begin{definition}[Venn $f$-diversity] \label{def:venn}
  Let $f: \{0, 1\}^r \rightarrow \mathbb{N}$. We define the \emph{Venn $f$-diversity} of a list $S_1, \ldots, S_r$ of subsets of $V$ to be
  \begin{equation*}
    Div_f(S_1, \ldots, S_r) \coloneqq \sum_{v \in V} I_f(S_1, \ldots, S_r,v).
  \end{equation*}  
\end{definition}

\begin{figure}
  \begin{center}
        \begin{tikzpicture}[scale=0.84, transform shape]
      \draw[thick] (-5,-4) rectangle (5,3)node[above right]{};
      \node[venn=-35, label=160:$A$] (A) at (0,0){};
      \node[venn=35, label=20:$B$] (B) at (0,0){};
      \node[venn=-30, label=270:$C$] (C) at (-1.5,-1){};
      \node[venn=30, label=270:$D$] (D) at (1.5,-1){};
      
      \node at ([shift=(-20:.9)]A.180){3}; %
      \node at ([shift=(200:.9)]B.0){3}; %
      \node at ([shift=(-35:1.1)]C.180){3}; %
      \node at ([shift=(215:1.1)]D.0){3}; %
      
      \node at ([shift=(90:.8)]A.center){4}; %
      \node at ([shift=(-72:1.25)]A.180){4}; %
      \node at ([shift=(140:.8)]A.0){4}; %
      \node at ([shift=(40:.8)]B.180){4}; %
      \node at ([shift=(252:1.25)]B.0){4}; %
      \node at (0.0,-2.4){4}; %
      
      \node at (-1.1,-.2){3}; %
      \node at (1.1,-.2){3}; %
      \node at (.8,-1.7){3}; %
      \node at (-.8,-1.7){3}; %

      \node at (0,-1){4};
      
      \node at (0,-3.5){0};
    \end{tikzpicture}  
  \end{center}

  \caption{A Venn diagram for four solutions: $A, B, C, D$.
    The influence of a vertex corresponding to the $DivSum$ measure is marked. It is the number of solutions it belongs to multiplied by the number of solutions it does not belong to.
  With Venn diversity measures, one could assign an arbitrary influence value to each part of this diagram.
  }\label{fig:venn}
\end{figure}

Whenever $f$ is fixed, we will write $I(S_1, \ldots, S_r, v)$ instead of $I_f(S_1, \ldots, S_r, v)$ and $Div(S_1, \ldots, S_r)$ instead of $Div_f(S_1, \ldots, S_r)$, and we say Venn diversity instead of Venn $f$-diversity.
We now give an example of a Venn $f$-diversity measure other than $DivSum$. 

\begin{example}[$Div^*$ measure]\label{ex:divstar}
  \begin{equation*}
    Div^*(S_1, \ldots, S_r) \coloneq \sum_{v \in V}(r^2 - |\{i \in [r]: v \in S_i\}|^2).
  \end{equation*}
\end{example}

Returning to our example with only two distinct solutions copied multiple times, achieving the best $DivSum$ value, we suggest that in that case, one might try to use the $Div^*$ measure defined in \cref{ex:divstar}.
Contrary to the sum of pairwise hamming distances, this measure penalizes taking many copies of two disjoint sets; see \cref{fig:ds} for an example of a small problem instance already admiting such a behavior.

\begin{figure}
  \begin{center}
    \begin{tikzpicture}[scale=0.85]
    \draw (0,3) node [circle, style=draw, fill=black] {} -- (1,3) node [circle, style=draw, fill=white] {} -- (2,3) node [circle, style=draw, fill=black] {} -- (3,3) node [circle, style=draw, fill=white] {} -- (4,3) node [circle, style=draw, fill=black] {} (5,3) node [] {A};
    \draw (0,2) node [circle, style=draw, fill=white] {} -- (1,2) node [circle, style=draw, fill=black] {} -- (2,2) node [circle, style=draw, fill=white] {} -- (3,2) node [circle, style=draw, fill=black] {} -- (4,2) node [circle, style=draw, fill=white] {} (5,2) node [] {B};
    \draw (0,1) node [circle, style=draw, fill=white] {} -- (1,1) node [circle, style=draw, fill=black] {} -- (2,1) node [circle, style=draw, fill=white] {} -- (3,1) node [circle, style=draw, fill=white] {} -- (4,1) node [circle, style=draw, fill=black] {} (5,1) node [] {C};     
    \draw (0,0) node [circle, style=draw, fill=black] {} -- (1,0) node [circle, style=draw, fill=white] {} -- (2,0) node [circle, style=draw, fill=white] {} -- (3,0) node [circle, style=draw, fill=black] {} -- (4,0) node [circle, style=draw, fill=white] {} (5,0) node [] {D};
  \end{tikzpicture}    
  \end{center}

  \caption{
    Consider the dominating set problem on a path of five vertices. 
If the number of required solutions is four, then a solution maximizing $DivSum$ will consist of twice A and twice B, having $DivSum$ equal to 20.
However, for the $Div^*$ function, the value for this set of solutions is equal to 60, which is less than for an arguably more diverse set of solutions consisting of all A, B, C, and D, with the $Div^*$ being equal to 63. 
  }\label{fig:ds}
\end{figure}

\paragraph{Organization of the paper.}

After short preliminaries (\Cref{sec:prelim}), 
we present the proof of our main theorem (\Cref{mainThm}) in \Cref{sec:mainThm}.
Then we provide an exposition of our main proof and tools on an example of the vertex cover problem (\Cref{sec:VC}).
We also give a simple example of a non-monotone dynamic program in \Cref{sec:nonMonotone}.
Then we include a section showing monotone DP for MSO$_1$ problems \Cref{sec:MSO}.
That section provides a proof of \Cref{cor:mainCor}.
We conclude with a discussion about open problems and further extensions in \Cref{sec:conclusions}.

\toappendix{
  \lv{
  \section{Preliminaries}\label{sec:prelim}
All graphs in this work are simple undirected graphs.
We keep most of the definitions consistent with these used in \cite{BFJMOPR22}.
We use $[r]$ to denote the set $\{1,\ldots,r\}$.

\subsection{Parameterized Complexity}
Let $\Sigma$ be a finite alphabet. A classical problem $\mathcal{P}$ is a subset of $\Sigma^*$.
Usually, $\mathcal{P}$ consists of proper encodings of some combinatorial objects, e.g., graphs (then we call $\mathcal{P}$ a \emph{graph problem}).
A \emph{parameterized problem} is a subset $\mathcal{P} \in \Sigma^* \times \mathbb{N}$. For an input $(x, k) \in \Sigma^* \times \mathbb{N}$,
the second part $k$ is called the \emph{parameter}.

\begin{definition}[FPT]
 An algorithm $\mathcal{A}$ is called \emph{fixed-parameter algorithm} for a parameterized problem $\mathcal{P}$,
 if there exists a computable function $f: \mathbb{N} \rightarrow \mathbb{N}$ and a constant $c$ such that
 for each input $(x, k) \in \Sigma^* \times \mathbb{N}$, $\mathcal{A}$ correctly decides whether $(x, k) \in \mathcal{P}$
 in time at most $f(k) \cdot |x|^c$.
 A parameterized problem $\mathcal{P}$ is \emph{fixed-parameter tractable} (FPT) if it admits a fixed-parameter algorithm.
\end{definition}

\begin{definition}[XP]
A parameterized problem $\mathcal{P}$ is called \emph{slice-wise polynomial} (XP)
if there exist an algorithm $\mathcal{A}$ and two computable functions $f, g: \mathbb{N} \rightarrow \mathbb{N}$ such that
for each input $(x, k) \in \Sigma^* \times \mathbb{N}$, $\mathcal{A}$ correctly decides whether $(x, k) \in \mathcal{P}$
in time at most $f(k) \cdot |x|^{g(k)}$.
\end{definition}

In instances where the parameter is relatively small, FPT and XP algorithms are efficient.
In the case of diverse problems, we are usually interested in finding a small set of diverse solutions,
so the number of requested solutions is a natural candidate for parameterization.

\subsection{Cliquewidth}\label{sec:cw}

\begin{definition}[Cliquewidth]\label{def:cw}
  Let $\omega \in \mathbb{N}$.
  A \emph{cliquewidth decomposition} of \emph{width} $\omega$ is a rooted tree $\mathcal{D}$ where each node $t \in V(\mathcal{D})$ is of one of the following types:
  \begin{itemize}
    \item $t = intro(v, a)$, where $a \in [\omega]$ and $t$ is a leaf of $\mathcal{D}$.
    This node corresponds to construing a one-vertex colored graph with a vertex $v$ of color $a$.
    \item $t = t_1 \oplus t_2$, where $t_1, t_2$ are children of $t$.
    This node corresponds to taking the disjoint union of colored graphs constructed by the subtrees of $t_1, t_2$.
    \item $t = recolor(t', a \rightarrow b)$, where $t'$ is the only child of $t$, $a, b \in [\omega]$.
    This node corresponds to recoloring all the vertices of color $a$ in a graph constructed by the subtree of $t'$ to color $b$.
    \item $t = addEdges(t', a, b)$, where $t'$ is the only child of $t$, $a, b \in [\omega]$.
    This node corresponds to adding an edge between each pair of vertices $u, v$ of colors $a, b$ respectively in a graph constructed by a subtree of $t'$ if such an edge does not exist.
  \end{itemize}
  For a node $t \in V(\mathcal{D})$ we denote by $\mathcal{D}_t$ the subtree of $\mathcal{D}$ rooted at $t$, by $G_t$ the colored graph constructed by $\mathcal{D}_t$, and by $\delta(t)$ the number of children of $t$.

  Let $G$ be a simple graph.
  $\mathcal{D}$ is a cliquewidth decomposition of $G$ if $G$ is isomorphic to $G_t$ for $t$ root of $\mathcal{D}$.
\end{definition}

Observe that in this definition, for each $v \in V(G)$, there is exactly one node $t = intro(v, a)$, where $a \in [\omega]$. We also write $t = Leaf_{\mathcal{D}}(v)$.
  The \emph{cliquewidth} of $G$ is minimal $\omega \in \mathbb{N}$ such that $G$ admits a decomposition $\mathcal{D}$ of width $\omega$.
}

\subsection{$DivSum$ is a Venn $f$-diversity Measure}\label{sec:measures}

Observe that we can rewrite the sum of pairwise hamming distances as
\begin{align*}
  \sum_{1 \leq i < j \leq r} HamDist(S_i, S_j)
  \sv{&}= \sum_{1 \leq i < j \leq r} \sum_{v \in V} \gamma(S_i, S_j, v) \sv{\\}
  \sv{&}= \sum_{v \in V} \sum_{1 \leq i < j \leq r}\gamma(S_i, S_j, v),
\end{align*}
where $\gamma(X, Y, v) = 1$ if $v \in X \setminus Y \cup Y \setminus X$, $0$ otherwise. Intuitively, $\gamma(X, Y, v)$ is $1$ if and only if the element $v$ contributes to the Hamming distance between $X$ and $Y$. Observe that
\begin{align*}
\sv{&}\sum_{1 \leq i < j \leq r}\gamma(S_i, S_j, v) \sv{\\} &= \sum_{1 \leq i < j \leq r}[v \in S_i \setminus S_j] + \sum_{1 \leq i < j \leq r}[v \in S_j \setminus S_i] \\
&= \sum_{1 \leq i < j \leq r}[v \in S_i \setminus S_j] + \sum_{1 \leq j < i \leq r}[v \in S_i \setminus S_j] \\
&= \sum_{1 \leq i, j \leq r}[v \in S_i \setminus S_j] \sv{\\&}= |\{i \in [r]: v \in S_i\}| \cdot |\{j \in [r]: v \not\in S_j\}| \sv{\\&}= I_f(S_1, \ldots, S_r, v),
\end{align*}
where ${f(m) = |\{i \in [r]: m[i] = 1\}| \cdot |\{j \in [r]: m[j] = 0\}|}$ for $m \in \{0, 1\}^r$.
And we have
\begin{equation*}
  DivSum(S_1, \ldots, S_r) = \sum_{v \in V} I_f(S_1, \ldots, S_r, v).
\end{equation*}

\lv{

\subsection{Graph Logic}
Graph problems can be described using logical formulas.
In the \emph{first-order logic} formula, only quantification over single variables is allowed, while in the \emph{second-order logic}, we allow quantification over predicates.
$MSO$ is the fragment of second-order logic where the second-order quantification is limited to unary predicates, which can also be viewed as subsets of the universe.
A graph can be modeled as a logical structure in several ways.
The universe may be a set of vertices, and the signature may consist of one binary relation, the adjacency. This model is called $MSO_1$.
Another way is to use vertices and edges as the universe, with the binary incidence relation. This model is denoted as $MSO_2$.
It can be shown that allowing quantification over sets of edges in $MSO_2$ makes this model more powerful than $MSO_1$.
For example, the existence of a Hamiltonian Path can be expressed in $MSO_2$, but not in $MSO_1$.
See, for example,~\cite{MSObook} for further details.
In this paper, we focus on the first model.

\begin{definition}[$MSO_1$ graph logic]
  An \emph{$MSO_1$ formula} for a graph problem $\mathcal{P}$ is a second order logic formula $\Phi$ (without free variables) that may contain:
  \begin{itemize}
    \item logical connectives,
    \item variables for vertices and vertex sets,
    \item predicates for membership, equivalence, and adjacency relation,
    \item constants, being particular vertices or vertex sets,
  \end{itemize}
  such that $G \in \mathcal{P}$ if and only if $G \models \Phi$.
  Similarly, an \emph{$MSO_1$ formula for a vertex problem} $\mathcal{P}$ is $\Phi$ with one free variable such that $(G, S) \in \mathcal{P}$ if and only if $G \models \Phi(S)$.
\end{definition}
}
}

\section{Monotone DP for Vertex Problems}\label{sec:mainThm}
A very natural property of dynamic programming algorithms for vertex problems on cliquewidth decompositions is that each partial solution for a node $t$ is a disjoint union of partial solutions for its children.
We call algorithms with this property \textit{monotone}.

In this section, we give the formalism of dynamic programming and the monotone property.
We prove our main result, that having a monotone dynamic program for each of the vertex problems $\mathcal{P}_1, \ldots, \mathcal{P}_r$ and a cliquewidth decomposition $\mathcal{D}$ of $G$, we can solve $Div^d(\mathcal{P}_1, \ldots, \mathcal{P}_r)$ on $G$.
\lv{\subsection{Definitions}}%
Recall that $\delta(t)$ denotes the number of children of a node $t \in V(\mathcal{D})$.
\begin{definition}[Dynamic programming core]
  A \emph{dynamic programming core} is an algorithm $\mathfrak{C}$ that takes cliquewidth decomposition $\mathcal{D}$ of a graph $G$ and produces
  \begin{itemize}
    \item A finite set $Accept_{\mathfrak{C}, \mathcal{D}} \subseteq \{0, 1\}^*$,
    \item A finite set $Process_{\mathfrak{C}, \mathcal{D}}(t) \subseteq (\{0, 1\}^*)^{\delta(t) + 1}$ for each $t \in V(\mathcal{D})$.
  \end{itemize}

  Let $\tau(\mathfrak{C}, \mathcal{D})$ be the time necessary to construct these data and $Size_{\mathfrak{C}, \mathcal{D}} \coloneqq max_{t \in V(\mathcal{D})}|Process_{\mathfrak{C, \mathcal{D}}}(t)|$.
\end{definition}

Intuitively, $Process_{\mathfrak{C}, \mathcal{D}}$ corresponds to the transitions of the dynamic programming.
If a partial solution represented by $w$ at node $t$
can be obtained from partial solutions represented by $w_1\ldots w_{\delta(t)}$ at its children $t_1\ldots t_{\delta(t)}$ then
$ww_1\ldots w_{\delta(t)} \in Process_{\mathfrak{C}, \mathcal{D}}(t)$.

\begin{definition}[Witness]
  Let $\mathcal{D}$ be a cliquewidth decomposition of a graph $G$ with a root $q$ and $\mathfrak{C}$ be a dynamic programming core.
  For $w \in \{0, 1\}^*$ a $(\mathfrak{C}, \mathcal{D}, w)$-witness is a function $\alpha: V(\mathcal{D}) \rightarrow  \{0, 1\}^*$ such that
  \begin{itemize}
    \item for each $t \in V(\mathcal{D})$
    with children $t_1, \ldots , t_{\delta(t)}$, 
    $\alpha(t)\alpha(t_1)\ldots \alpha(t_{\delta(t)}) \in Process_{\mathfrak{C}, \mathcal{D}}(t)$,
    \item $\alpha(q) = w$.
  \end{itemize}    

  If $w \in Accept_{\mathfrak{C}, \mathcal{D}}$, we say $\alpha$ is a $(\mathfrak{C}, \mathcal{D})$-witness.
\end{definition}
Intuitively, if $w$ is a dynamic programming entry in the root $t$ of $\mathcal{D}$, it can originate from many combinations of entries in the children of $t$.
A $(\mathfrak{C}, \mathcal{D}, w)$-witness describes one of many derivation trees for~$w$.

Observe that if $\mathcal{P}$ is a vertex problem, given a graph $G$, we cannot hope to list all $S$ such that $(G, S) \in \mathcal{P}$ in FPT time,
as the number of such sets may be exponential in $|G|$. Therefore, we focus on deciding whether for a given graph $G$ there exists any such $S$.
We call this decision problem a corresponding \emph{graph problem} and write $G \in \mathcal{P}$ if the answer is \textsc{Yes}.

\begin{definition} 
  We say that a dynamic programming core $\mathfrak{C}$ \emph{solves} a graph problem $\mathcal{P}$ if for each graph $G$ and for each cliquewidth decomposition $\mathcal{D}$ of $G$, $G \in \mathcal{P}$ if and only if a $(\mathfrak{C}, \mathcal{D})$-witness exists.
\end{definition}

That means $G$ is a \textsc{Yes}-instance of $\mathcal{P}$ if and only if there exists a derivation tree for some accepting entry~$w$.

\begin{theorem} \label{solves}
  Let $\mathcal{P}$ be a graph problem and $\mathfrak{C}$ be a dynamic programming core that solves $\mathcal{P}$.
  Given a decomposition $\mathcal{D}$ of a graph $G$ we can determine whether $G \in \mathcal{P}$ in time 
  \begin{equation*}
    \mathcal{O}\left(\tau(\mathfrak{C}, \mathcal{D}) + |V(\mathcal{D})| \cdot Size_{\mathfrak{C}, \mathcal{D}}\right).
  \end{equation*}
\end{theorem}

\begin{proof}
  Given $\mathfrak{C}, \mathcal{D}$, we can construct $Accept_{\mathfrak{C}, \mathcal{D}}$ and
  $Process_{\mathfrak{C}, \mathcal{D}}(t)$ for each $t \in V(\mathcal{D})$ in time $\tau(\mathfrak{C}, \mathcal{D})$. 
  Let $\Pi(t)$ be the set of $w \in \{0, 1\}^*$ such that a $(\mathfrak{C}, \mathcal{D}_t, w)$-witness exists.
  We can construct $\Pi(t)$ for each $t \in V(\mathcal{D})$ by induction.
  \begin{itemize}
    \item If $t$ is a leaf of the decomposition, $\Pi(t) = Process_{\mathfrak{C}, \mathcal{D}}(t)$.
    \item If $t$ has children $t_1, \ldots t_{\delta(t)}$,
    $\Pi(t) = \{w: \exists_{(w, w_1, \ldots w_{\delta(t)}) \in Process_{\mathfrak{C}, \mathcal{D}}(t)} \ s.t. \
    \forall_{i \in [\delta(t)]} \ w_i \in \Pi(t_i)\}$.
  \end{itemize}
  $G \in \mathcal{P}$ if and only if $\Pi(q) \cap Accept_{\mathfrak{C}, \mathcal{D}} \not = \emptyset$, where $q$ is the root of $\mathcal{D}$.
\end{proof}
Note that from the above algorithm we can also easily recover some $(\mathfrak{C}, \mathcal{D}, w)$-witness for each $w \in \Pi(q)$.

\begin{definition}[Monotonicity]\label{def:monotone}
  Let $\mathcal{P}$ be a vertex problem. We call a dynamic programming core $\mathfrak{C}$ \emph{monotone} if there exists a function
  $\rho: \{0,1\}^* \rightarrow \{0,1\}$ such that for each graph $G$, for each cliquewidth decomposition $\mathcal{D}$ of $G$, for each
  $S \subseteq V(G)$, $(G, S) \in \mathcal{P}$ if and only if there exists a $(\mathfrak{C}, \mathcal{D})$-witness $\alpha$ such that
  ${S = \{v \in V(G): \rho(\alpha(Leaf_\mathcal{D}(v))) = 1\}}$. We call such $\rho$ a \emph{$\mathfrak{C}$-vertex-membership function} and say \lv{a }pair $(\mathfrak{C}, \rho)$ \emph{solves}~$\mathcal{P}$.
\end{definition}

Intuitively, a vertex membership function $\rho$ together with a cliquewidth decomposition $\mathcal{D}$ of a graph $G$ and a $(\mathfrak{C}, \mathcal{D})$-witness $\alpha$ provide an encoding of a subset of vertices of $G$ using only the values of DP entries in leaves of $\mathcal{D}$. %
Let $S_\rho(\mathcal{D}, \alpha) = \{v \in V(G): \rho(\alpha(Leaf_\mathcal{D}(v))) = 1\}$.
Note that in this definition, we determine whether a vertex is in the subset or not based only on the value of the witness on its introduce node.

\lv{
\subsection{Proof of \Cref{mainThm}}
}

The following theorem is a formal reformulation of \Cref{mainThm} stated in the introduction. 
Indeed, if $t_\mathcal D= \mathcal{O}(Size_{\mathfrak{C}_i, \mathcal{D}})$ for $\mathcal{P}_i$ being the one with the slowest maximal running time on a single node of the decomposition out of $\mathcal{P}_1, \ldots, \mathcal{P}_r$,
and assuming we are given $Accept_{\mathfrak{C}_i, \mathcal{D}}$ and $Process_{\mathfrak{C}_i, \mathcal{D}}(t)$ for each $t \in V(\mathcal{D})$ precomputed, the time needed to solve $Div^d(\mathcal{P}_1,  \ldots , \mathcal{P}_r)$ is
\[
  \mathcal{O}\left(|V(\mathcal{D})| \cdot \prod_{i = 1}^r Size_{\mathfrak{C}_i, \mathcal{D}}\right) \le \mathcal O(|V(\mathcal{D})|t_{\mathcal{D}}^r).
  \]

\begin{theorem} \label{theorem}
  Let $\mathcal{P}_1, \ldots , \mathcal{P}_r$ be vertex problems, $(\mathfrak{C}_i, \rho_i)$
  be a pair of a monotone dynamic programming core and a vertex membership function that solves $\mathcal{P}_i$,
  $Div$ be a Venn $f$-diversity measure for some ${f: \{0, 1\}^r \rightarrow \mathbb{N}}$, $d \in \mathbb{N}$.
  The problem $Div^d(\mathcal{P}_1,  \ldots , \mathcal{P}_r)$ on a graph $G$ with given cliquewidth decomposition $\mathcal{D}$ can be solved in time
  \begin{equation*}
    \mathcal{O}\left(\sum_{i = 1}^{r}{\tau(\mathfrak{C}_i, \mathcal{D})} + |V(\mathcal{D})| \cdot \prod_{i = 1}^r Size_{\mathfrak{C}_i, \mathcal{D}}\right).
  \end{equation*}
\end{theorem}

Before proving it, we will also state the following theorem, which is a formal reformulation of \cref{minmainThm} from the introduction.
The proof of \cref{theorem} can be easily modified to work for \cref{thm:min-diversity}.

\begin{theorem} \label{thm:min-diversity}
  Let $\mathcal{P}_1,  \ldots , \mathcal{P}_r$ be vertex problems, $(\mathfrak{C}_i, \rho_i)$
  be a pair of a monotone dynamic programming core and a vertex membership function that solves $\mathcal{P}_i$,
  $d \in \mathbb{N}$.
  The problem $DivMin^d(\mathcal{P}_1,  \ldots , \mathcal{P}_r)$ on a graph $G$ with given cliquewidth decomposition $\mathcal{D}$ can be solved in time
  \begin{equation*}
    \mathcal{O}\left(\sum_{i = 1}^{r}{\tau(\mathfrak{C}_i, \mathcal{D})} + |V(\mathcal{D})| \cdot \prod_{i = 1}^r Size_{\mathfrak{C}_i, \mathcal{D}} \cdot d^{r^2}\right).
  \end{equation*}
\end{theorem}

Keep in mind that the following proof can be modified to work also for $DivMin$ as stated in \cref{thm:min-diversity}.
Instead of one diversity value $\ell$ for a tuple of partial solutions, we need to store a vector $L \in \{0, \ldots, d\} ^ {\binom{r}{2}}$ of pairwise Hamming distances between the solutions.
Since for $S_1, S_2 \subseteq S, \ P_1, P_2 \subseteq P$, where $S$ and $P$ are disjoint, $HamDist(S_1 \cup P_1, S_2 \cup P_2) = HamDist(S_1, S_2) + HamDist(P_1, P_2)$, we can sum these vectors point-wise as we sum single diversity values.
As the Hamming distance only grows, and we are interested in having the minimum distance of at least $d$, we can treat all the values greater than $d$ as equal to $d$.
However, we cannot take advantage of keeping only the best diversity value for a tuple of entries because such vectors may be incomparable, so for one tuple of entries that represents multiple tuples of partial solutions, there may be up to $(d + 1)^{\binom{r}{2}}$ different diversity vectors, which contributes to $d^{r^2}$ factor in the running time of \Cref{thm:min-diversity}.

  The correctness of the algorithm is based on the following lemma and the fact that for each node $t = t_1 \oplus t_2$  of the decomposition, $V(G_{t_1}) \cap V(G_{t_2}) = \emptyset$.

  \begin{restatable}{lemma}{reslemma} \label{lemma}
    Let $S, P \subseteq V$ be disjoint sets, $S_1, \ldots, S_r \subseteq S$, $P_1, \ldots, P_r \subseteq P$ and
    let $Div$ be a Venn $f$-diversity measure for some $f: \{0, 1\}^r \rightarrow \mathbb{N}$. Then
    \begin{align*}
      Div(S_1 \cup P_1,\ldots, S_r \cup P_r) \sv{\\ &}=Div(S_1, \ldots, S_r) + Div(P_1, \ldots, P_r) - |V(G)| \cdot f(0^r).
    \end{align*}
  \end{restatable}
  
  Note that in \Cref{lemma}, we need to know the number of vertices of the final graph ($|V(G)|$).
However, $f(0^r)$ corresponds to the influence of a vertex that is in none of the partial solutions, which could be set to $0$ in many convenient measures.
In such cases, we can even ignore this mild constraint.

  \begin{proof}[Proof of \cref{lemma}]
    Recall that for $S_1, \ldots, S_r \subseteq V$
    \begin{equation*}
      Div(S_1, \ldots S_r) = \sum_{v \in V} f(m(S_1, \ldots, S_r, v)),
    \end{equation*}
    where $m(S_1, \ldots, S_r, v)$ is the membership vector of $v$ in $S_1, \ldots, S_r$.
    If $v \not\in S$, $m(S_1, \ldots S_r, v) = 0^r$, and analogously if $v \not\in P$, $m(P_1, \ldots P_r, v) = 0^r$ so we get
    \begin{align*}
      \sv{&}Div(S_1, \ldots, S_r) = \sum_{v \in V} f(m(S_1, \ldots, S_r, v)) \sv{\\&} =\sum_{v \in S} f(m(S_1, \ldots, S_r, v)) + |P| f(0^r) + |V \setminus (S \cup P)| f(0^r), \\
      \sv{&}Div(P_1, \ldots, P_r) = \sum_{v \in V} f(m(P_1, \ldots, P_r, v)) \sv{\\&} =\sum_{v \in P} f(m(P_1, \ldots, P_r, v)) + |S| f(0^r) + |V \setminus (S \cup P)| f(0^r).
    \end{align*}
    Observe that
    \begin{itemize}
      \item if $v \in S$, $v \in S_i \cup P_i$ if and only if $v \in S_i$,
        so ${m(S_1 \cup P_1, \ldots, S_r \cup P_r, v) = m(S_1, \ldots, S_r, v)}$,
      \item if $v \in P$, $v \in S_i \cup P_i$ if and only if $v \in P_i$,
        so ${m(S_1 \cup P_1, \ldots, S_r \cup P_r, v) = m(P_1,\ldots, P_r, v)}$,
      \item if $v \in V \setminus (S \cup P)$, $v \not\in S_i$ and $v \not\in P_i$ for each $i \in [r]$,
        so $m(S_1 \cup P_1, \ldots, S_r \cup P_r, v) = 0^r$.      
    \end{itemize}
    So we get
    \begin{align*}
      \sv{&}Div(S_1 \cup P_1,\ldots, S_r \cup P_r) \sv{\\}
      &= \sum_{v \in V} f(m(S_1 \cup P_1, \ldots, S_r \cup P_r, v)\\
      &= \sum_{v \in S} f(m(S_1, \ldots, S_r, v)) + \sum_{v \in P} f(m(P_1,\ldots, P_r, v)) \sv{+\\&}+ |V \setminus (S \cup P)| \cdot f(0^r) \\
      &= Div(S_1, \ldots, S_r) + Div(P_1, \ldots, P_r) - |V| \cdot f(0^r).\qedhere
    \end{align*}
  \end{proof}
  We start by computing the tables $Process_{\mathfrak{C}_i, \mathcal{D}}$ and $Accept_{\mathfrak{C}_i, \mathcal{D}}$ corresponding to each core in total time $\sum_{i = 1}^{r} \tau(\mathfrak{C}_i, \mathcal{D})$.
  Then, for each node $t$ of the decomposition, we construct by induction the set $\Pi_{Div}(t)$ of tuples $(w^1, \ldots, w^r, \ell)$ such that, for each $i$, there exists a $(\mathfrak{C}_i, \mathcal{D}_t, w^i)$-witness $\alpha_i$, and $\ell$ is the diversity of the list of partial solutions defined by this tuple of witnesses. If for some $(w^1, \ldots, w^r)$ there are multiple such tuples of witnesses, we keep only those that give the maximum value of the diversity. If $t$ is a leaf of $\mathcal{D}$, each witness $\alpha_i$ consists only of $w^i \in Process_{\mathfrak{C}_i, \mathcal{D}}$. Observe that for any tuple $w = (w^1, \ldots, w^r)$, the only vertex whose membership in the list of partial solutions defined by this tuple can be nonzero is the vertex $v$ introduced at $t$. Moreover, we can compute the corresponding membership vector $m_w$ by simply applying the membership function. Then $\Pi_{Div}(t)$ consists of all such tuples $w$, together with the diversity value $\ell = f(m_w) + (|V(G)| - 1)f(0^r)$. If $t$ has one child $t_1$, let $\Pi_{Div}(t)$ be the set of tuples $(w^1, \ldots, w^r, \ell)$ such that $\ell$ is the maximal value of $\ell_1$ among all $(w_1^1, \ldots, w_1^r, \ell_1) \in \Pi_{Div}(t_1)$ where each $(w^i, w^i_1) \in Process_{\mathfrak{C}_i, \mathcal{D}}(t)$. If $t$ has children $t_1, t_2$, let $\Pi_{Div}(t)$ be the set of tuples $(w^1, \ldots, w^r, \ell)$ such that
  $\ell$ is the maximal value of $\ell_1 + \ell_2 - |V(G)| f(0^r)$ among all pairs of $(w_1^1, \ldots, w_1^r, \ell_1) \in \Pi_{Div}(t_1), (w_2^1, \ldots, w_2^r, \ell_2) \in \Pi_{Div}(t_2)$ where each $(w^i, w^i_1, w^i_2) \in Process_{\mathfrak{C}_i, \mathcal{D}}(t)$. Note that we are using the monotonicity assumption in this argument. As we are joining pairs of partial solutions belonging to disjoint sets $V(G_{t_1}), V(G_{t_2})$, we can compute the diversity using \cref{lemma}. Observe that for each node $t$ the set $\Pi_{Div}(t)$ can be computed in time $\prod_{i = 1}^{r}|Process_{\mathfrak{C}_i, \mathcal{D}}(t)|$. The last step is checking whether there exists $(w^1, \ldots , w^r, \ell) \in \Pi_{Div}(q)$ such that $\ell \geq d$ and each $w^i \in Accept_{\mathfrak{C}_i, \mathcal{D}}$ for $q$ being the root of $\mathcal{D}$.
This can be done in $\prod_{i = 1}^{r}|Accept_{\mathfrak{C}_i, \mathcal{D}}|$ time.

We prove the correctness by combining the two following lemmas.
  Firstly, %
  we show that for each $S_1, \ldots, S_r$ such that $(G, S_i) \in \mathcal{P}_i$ there exists
  $(w^1, \ldots, w^r, \ell) \in \Pi_{Div}(q)$ for some $\ell \geq Div(S_1, \ldots, S_r)$.
 
  \begin{restatable}{lemma}{reslowerbound} \label{lowerbound}
    Let $r$ be an integer, $\mathcal D$ be a decomposition tree, $t \in V(\mathcal{D})$, $Div$ be a Venn $f$-diversity measure for some $f: \{0, 1\}^r \rightarrow \mathbb{N}$,
    and $w^1, \ldots, w^r \in \{0,1\}^*$.
    If for each $i \in[r]$ there exists a $(\mathfrak{C}_i, \mathcal{D}_t, w^i)$-witness $\alpha_i$, then $(w^1, \ldots, w^r, \ell) \in \Pi_{Div}(t)$ for some \[\ell \geq Div(S_{\rho_1}(\mathcal{D}_t, \alpha_1), \ldots, S_{\rho_r}(\mathcal{D}_t, \alpha_r)).\]
  \end{restatable}

  Secondly,  %
  we show that for each $(w^1, \ldots, w^r, \ell) \in \Pi_{Div}(q)$ there exists $S_1, \ldots, S_r$
  such that $(G, S_i) \in \mathcal{P}_i$ and ${Div(S_1, \ldots, S_r) = \ell}$. 

  \begin{restatable}{lemma}{resupperbound} \label{upperbound}
    Let $r$ be an integer, $\mathcal{D}$ be a decomposition tree, $t \in V(\mathcal{D})$, $Div$ be a Venn $f$-diversity measure for some $f: \{0, 1\}^r \rightarrow \mathbb{N}$,
    and $w^1, \ldots, w^r \in \{0,1\}^*$.
    If $(w^1, \ldots, w^r, \ell) \in \Pi_{Div}(t)$, then for each $i \in [r]$
    there exists a $(\mathfrak{C}_i, \mathcal{D}_t, w^i)$-witness $\alpha_i$ such that
    \[Div(S_{\rho_1}(\mathcal{D}_t, \alpha_1), \ldots, S_{\rho_r}(\mathcal{D}_t, \alpha_r)) = \ell.\]
  \end{restatable}

  We prove \cref{theorem} before proving \cref{lowerbound} and \cref{upperbound}.

\begin{proof}[Proof of \cref{theorem}]
  We will solve $Div^d(\mathcal{P}_1,  \ldots , \mathcal{P}_r)$ by the following steps:
  \begin{enumerate}
    \item
  First,  constructing the data corresponding to each core $\mathfrak{C}_i$.
    This step takes $\sum_{i = 1}^{r} \tau(\mathfrak{C}_i, \mathcal{D})$ time.

  \item
    Second, constructing by induction for each $t \in V(\mathcal{D})$ set $\Pi_{Div}(t)$.
    \begin{itemize}
      \item
      If $t = Leaf_{\mathcal{D}}(v)$,
      \begin{align*}
        \Pi_{Div}(t) = \{(w^1, \ldots , w^r, \ell)\lv{&} : \forall_{i \in [r]} w^i \in Process_{\mathfrak{C}_i, \mathcal{D}}(t),
        \\&\ell = f(\rho_1(w^1), \ldots, \rho_r(w^r)) + (|V(G)| - 1)f(0^r)
        \}.
      \end{align*}
      This can be computed in time $\prod_{i = 1}^r |Process_{\mathfrak{C}_i, \mathcal{D}}(t)|$.
     \item
     If a node $t$ has children $t_1, \ldots, t_{\delta(t)}$.
      \begin{itemize}
        \item If $\delta(t) = 1$,
        \begin{align*}
          \Pi_{Div}(t) = \{(w^1, \ldots , w^r, \ell) :
          \exists_{(w^1_1, \ldots w^r_1, \ell_1) \in \Pi_{Div}(t_1)} 
          \ s.t. \ \forall_{i \in [r]} (w^i, w^i_1) \in Process_{\mathfrak{C}_i, \mathcal{D}}(t), \\
           \ell = max\{\ell_1 : \exists_{(w^1_1, \ldots w^r_1, \ell_1) \in \Pi_{Div}(t_1)} 
          \ s.t. \ \forall_{i \in [r]} (w^i, w^i_1) \in Process_{\mathfrak{C}_i, \mathcal{D}}(t)\}\}.
        \end{align*}
        \item If $\delta(t) = 2$,
        \begin{align*}
          \Pi_{Div}(t) = \{&(w^1, \ldots , w^r, \ell) : 
          \exists_{(w^1_1, \ldots w^r_1, \ell_1) \in \Pi_{Div}(t_1), (w^1_{2}, \ldots w^r_{2}, \ell_{2}) \in \Pi_{Div}(t_{2})} \\
          & s.t. \ \forall_{i \in [r]} (w^i, w^i_1, w^i_{2}) \in Process_{\mathfrak{C}_i, \mathcal{D}}(t), 
           \ell = max\{\ell_1 + \ell_{2} - |V(G)| \cdot f(0^r) : \\
          & \exists_{(w^1_1, \ldots w^r_1, \ell_1) \in \Pi_{Div}(t_1), (w^1_{2}, \ldots w^r_{2}, \ell_{2}) \in \Pi_{Div}(t_{2})} 
           \ s.t. \ \forall_{i \in [r]} (w^i, w^i_1, w^i_{2}) \in Process_{\mathfrak{C}_i, \mathcal{D}}(t)\}
          \}.
        \end{align*}
      \end{itemize}
      This can also be computed in time $\prod_{i = 1}^r |Process_{\mathfrak{C}_i, \mathcal{D}}(t)|$ by iterating through all tuples
      \[((w^1, w^1_1, \ldots, w^1_{\delta(t)}), \ldots, (w^r, w^r_1, \ldots, w^r_{\delta(t)})) \in Process_{\mathfrak{C}_1, \mathcal{D}} \times \cdots \times Process_{\mathfrak{C}_r, \mathcal{D}}\]
      and checking whether for each $j \in [\delta(t)]$ there exists $(w^1_j, \ldots w^r_j, \ell_j) \in \Pi_{Div}(t_j)$ for some $\ell_j \in \mathbb{N}$.
    \end{itemize}
    This step takes $\sum_{t \in V(\mathcal{D})} \prod_{i = 1}^{r}|Process_{\mathfrak{C}_i, \mathcal{D}}(t)| = \mathcal{O}(|V(\mathcal{D})| \cdot \prod_{i = 1}^{r} Size_{\mathfrak{C}_i, \mathcal{D}})$ time.\\
    Observe that $(w^1, \ldots, w^r, \ell) \in \Pi_{Div}(t)$ if and only if
    $w^i \in \Pi_i(t)$ for each $i \in [r]$, where $\Pi_i(t)$ is the set of $w \in \{0, 1\}^*$ such that a $(\mathfrak{C}_i, \mathcal{D}_t, w)$-witness exists
    (see the proof of \cref{solves}).
    
    \item
    The last step is checking whether there exists $(w^1, \ldots , w^r, \ell) \in \Pi_{Div}(q)$,
    such that $\ell \geq d$ and $w^i \in Accept_{\mathfrak{C}_i, \mathcal{D}}$ for each $i \in [r]$, where $q$ is the root of $\mathcal{D}$,
    which takes $\mathcal{O}(\prod_{i = 1}^{r} Size_{\mathfrak{C}_i, \mathcal{D}})$ time.
  \end{enumerate}
  
  It remains to prove the correctness.
  Firstly, by \cref{lowerbound}, we have that for each $S_1, \ldots, S_r$ such that $(G, S_i) \in \mathcal{P}_i$ there exists
  $(w^1, \ldots, w^r, \ell) \in \Pi_{Div}(q)$ for some $\ell \geq Div(S_1, \ldots, S_r)$.
  Secondly, by \cref{upperbound}, we have that for each $(w^1, \ldots, w^r, \ell) \in \Pi_{Div}(q)$ there exists $S_1, \ldots, S_r$
  such that $(G, S_i) \in \mathcal{P}_i$ and ${Div(S_1, \ldots, S_r) = \ell}$.
  As a result, we get the following corollary.
  \begin{corollary}
    $(G, S_1, \ldots, S_r) \in Div^d(\mathcal{P}_1, \ldots, \mathcal{P}_r)$ if and only if there exists
    $(w_1,\ldots, w_r, \ell) \in \Pi_{Div}(q)$ such that ${\ell \geq d}$ and $w^i \in Accept_{\mathfrak{C}_i, \mathcal{D}}$ for each $i \in [r]$,
    where $q$ is the root of $\mathcal{D}$.
  \end{corollary}
  This finishes the proof of \Cref{theorem}.
\end{proof}

  \begin{proof}[Proof of \cref{lowerbound}]
  We use induction bottom-up on the decomposition tree $\mathcal D$.
    \paragraph{\textbf{Node $\mathbf{t = Leaf_{\mathcal{D}}(v)}$.}} %
      Assume that %
      there exists a $(\mathfrak{C}_i, \mathcal{D}_t, w^i)$-witness $\alpha_i$
      for each $i \in [r]$.
      Then by the definition of witness, $\alpha_i(t) = w^i$ and $\alpha_i(t) \in Process_{\mathfrak{C}_i, \mathcal{D}}(t)$, so $w_i \in Process_{\mathfrak{C}_i, \mathcal{D}}(t)$ for each $i \in[r]$.
      Then by the definition of $\Pi_{Div}$, $(w^1, \ldots, w^r, \ell) \in \Pi_{Div}(t)$, where $\ell = f(\rho_1(w^1), \ldots, \rho_r(w^r)) + (|V(G)| - 1)f(0^r)$.
      Observe that
      \begin{equation*}
        S_{\rho_i}(\mathcal{D}_t, \alpha_i) = \begin{cases}
          \{v\}  & \text{if $\rho_i(w^i) = 1$}, \\
          \emptyset & \text{otherwise}.
        \end{cases}
      \end{equation*}
      Then for each $u \in V(G)$,
      \begin{align*}
        m(S_{\rho_1}(\mathcal{D}_t, \alpha_1), \ldots, S_{\rho_r}(\mathcal{D}_t, \alpha_r), u)  = \begin{cases}
          (\rho_1(w^1), \ldots, \rho_r(w^r))  & \text{if $u = v$}, \\
          0^r & \text{otherwise}.
        \end{cases}
      \end{align*}
      Therefore, we get
      \begin{align*}
        Div(S_{\rho_1}(\mathcal{D}_t, \alpha_1), \ldots, S_{\rho_r}(\mathcal{D}_t, \alpha_r)) &= \sum_{u \in V(G)} f(m(S_{\rho_1}(\mathcal{D}_t, \alpha_1), \ldots, S_{\rho_r}(\mathcal{D}_t, \alpha_r), u)) \\
        &= f(\rho_1(w^1), \ldots, \rho_r(w^r)) + (|V(G)| - 1)f(0^r) = \ell.
      \end{align*}
      \noindent\textbf{Node $\mathbf{t}$ has children $\mathbf{t_1, \ldots t_{\delta(t)}}$.} %
      Assume that 
      for each $i \in [r]$ there exists a $(\mathfrak{C}_i, \mathcal{D}_t, w^i)$-witness $\alpha_i$ and let $w^i = \alpha_i(t)$.
      Let $t_j$ be a child of $t$. For $i \in [r]$ let $w_j^i \coloneq \alpha_i(t_j)$.
      Observe that $\alpha_i^j \coloneq \alpha_i |_{V(\mathcal{D}_{t_j})}$ is a $(\mathfrak{C}_i, \mathcal{D}_{t_j}, w^i_j)$-witness.
      Then, by the induction assumption, $(w^1_j, \ldots, w^r_j, l_j) \in \Pi_{Div}(t_j)$ for some $l_j \geq Div(S_{\rho_1}(\mathcal{D}_{t_j}, \alpha_1^j), \ldots, S_{\rho_r}(\mathcal{D}_{t_j}, \alpha_r^j))$.
      Also note that for each $i \in [r]$, $(w^i, w^i_1, \ldots, w^i_{\delta(t)}) \in Process_{\mathfrak{C}_i, \mathcal{D}}(t)$, as $\alpha_i$ is a $(\mathfrak{C}_i, \mathcal{D}_t, w^i)$-witness.

      If $\delta(t) = 1$, by the definition of $\Pi_{Div}$,
      ${(w^1, \ldots, w^r, \ell) \in \Pi_{Div}(t)}$ for some $\ell \geq \ell_1$.
      As for each $i \in [r]$ $S_{\rho_i}(\mathcal{D}_t, \alpha_i) = S_{\rho_i}(\mathcal{D}_{t_1}, \alpha_i^1)$, we get
      \begin{align*}
        Div(S_{\rho_1}(\mathcal{D}_t, \alpha_1), \ldots, S_{\rho_r}(\mathcal{D}_t, \alpha_r))
         = Div(S_{\rho_1}(\mathcal{D}_{t_1}, \alpha_1^1), \ldots, S_{\rho_r}(\mathcal{D}_{t_1}, \alpha_r^1))
        \leq \ell_1 \leq \ell.
      \end{align*}
      If $\delta(t) = 2$, by the definition of $\Pi_{Div}$, ${(w^1, \ldots, w^r, \ell) \in \Pi_{Div}(t)}$ for some $\ell \geq \ell_1 + \ell_2 - |V(G)|f(0^r)$.
      Observe that for each $i \in[r]$,
 $       S_{\rho_i}(\mathcal{D}_t, \alpha_i) = S_{\rho_i}(\mathcal{D}_{t_1}, \alpha_i^1) \cup S_{\rho_i}(\mathcal{D}_{t_2}, \alpha_i^2),$
      where $S_{\rho_1}(\mathcal{D}_{t_1}, \alpha_1^1), \ldots, S_{\rho_r}(\mathcal{D}_{t_1}, \alpha_r^1) \subseteq V(G_{t_1})$,
      $S_{\rho_1}(\mathcal{D}_{t_2}, \alpha_1^2), \ldots, S_{\rho_r}(\mathcal{D}_{t_2}, \alpha_r^2) \subseteq V(G_{t_2})$ and
      $V(G_{t_1}) \cap V(G_{t_2}) = \emptyset$. So by \cref{lemma}
      \begin{align*}
        &Div(S_{\rho_1}(\mathcal{D}_t, \alpha_1), \ldots, S_{\rho_r}(\mathcal{D}_t, \alpha_r)) \\
        &= Div(S_{\rho_1}(\mathcal{D}_{t_1}, \alpha_1^1), \ldots, S_{\rho_r}(\mathcal{D}_{t_1}, \alpha_r^1)) 
        + Div(S_{\rho_1}(\mathcal{D}_{t_2}, \alpha_1^2), \ldots, S_{\rho_r}(\mathcal{D}_{t_2}, \alpha_r^2)) - |V(G)|f(0^r) \\ 
        &\leq \ell_1 + \ell_2 - |V(G)|f(0^r) \leq \ell.
      \end{align*}
      It follows that $(w^1, \ldots, w^r, \ell) \in \Pi_{Div}(t)$ for some ${\ell \geq Div(S_{\rho_1}(\mathcal{D}_t, \alpha_1), \ldots, S_{\rho_r}(\mathcal{D}_t, \alpha_r))}$.
  \end{proof}

  \begin{proof}[Proof of \cref{upperbound}]
    Again we use induction bottom-up on the decomposition tree.
    \paragraph{\textbf{Node $\mathbf{t = Leaf_{\mathcal{D}}(v)}$.}}
      Assume $(w^1, \ldots, w^r, \ell) \in \Pi_{Div}(t)$.
      By the definition of $\Pi_{Div}$,
      \[\ell = f(\rho_1(w^1), \ldots, \rho_r(w^r)) + (|V(G)| - 1)f(0^r).\]
      Observe that for each $i \in [r]$ there exists a $(\mathfrak{C}_i, \mathcal{D}_t, w^i)$-witness $\alpha_i$ which assigns $w^i$ to $t$.
      Then
      \begin{equation*}
        S_{\rho_i}(\mathcal{D}_t, \alpha_i) = \begin{cases}
          \{v\}  & \text{if $\rho_i(w^i) = 1$,} \\
          \emptyset & \text{otherwise.}
        \end{cases}
      \end{equation*}
      Then for each $u \in V(G)$
      \begin{align*}
        m(S_{\rho_1}(\mathcal{D}_t, \alpha_1), \ldots, S_{\rho_r}(\mathcal{D}_t, \alpha_r), u)  = \begin{cases}
          (\rho_1(w^1), \ldots, \rho_r(w^r))  & \text{if $u = v$}, \\
          0^r & \text{otherwise}.
        \end{cases}
      \end{align*}
      Therefore, we get
      \begin{align*}
      Div(S_{\rho_1}(\mathcal{D}_t, \alpha_1), \ldots, S_{\rho_r}(\mathcal{D}_t, \alpha_r)) 
      &=\sum_{u \in V(G)} f(m(S_{\rho_1}(\mathcal{D}_t, \alpha_1), \ldots, S_{\rho_r}(\mathcal{D}_t, \alpha_r), u)) \\
      &= f(\rho_1(w^1), \ldots, \rho_r(w^r)) + (|V(G)| - 1)f(0^r) = \ell.
      \end{align*}
      \noindent\textbf{Node $\mathbf{t}$ has children $\mathbf{t_1, \ldots t_{\delta(t)}}$.}
      By the definition of $\Pi_{Div}$, if $(w^1, w^2, \ldots , w^r, \ell) \in \Pi_{Div}(t)$,
      there exists
      $(w^1_1, \ldots w^r_1, \ell_1) \in \Pi_{Div}(t_1), \ldots, (w^1_{\delta(t)}, \ldots w^r_{\delta(t)}, \ell_{\delta(t)}) \in \Pi_{Div}(t_{\delta(t)})$
      such that for each $i \in [r]$, $(w^i, w^i_1, \ldots, w^i_{\delta(t)})\allowbreak \in Process_{\mathfrak{C}_i, \mathcal{D}}$ and $\ell = \ell_1$ if $\delta(t) = 1$, $\ell = {\ell_1 + \ell_2 - |V(G)|f(0^r)}$ if $\delta(t) = 2$.
      Let $t_j$ be a child of $t$. By the induction assumption,
      if $(w^1_j, \ldots, w^r_j, \ell_j) \in \Pi_{Div}(t_j)$, then for each $i \in [r]$
      there exists a $(\mathfrak{C}_i, \mathcal{D}_{t_j}, w^i_j)$-witness $\alpha_i^j$ such that
      $Div(S_{\rho_1}(\mathcal{D}_{t_j}, \alpha_1^j), \allowbreak \ldots, S_{\rho_r}(\mathcal{D}_{t_j}, \alpha_r^j)) = \ell_j.$      
      For each $i \in [r], \ t' \in V(\mathcal{D}_t)$, set
      \[
        \alpha_i(t') \coloneq \begin{cases}
          \alpha_i^j(t')  & \text{if $t' \in V(\mathcal{D}_{t_j}), \ j \in \delta(t)$,} \\
          w^i & \text{if $t' = t$.}
        \end{cases}  
      \]
      Observe that $\alpha_i$ is a properly defined $(\mathfrak{C}_i, \mathcal{D}_t, w^i)$-witness. It remains to show that \[{Div(S_{\rho_1}(\mathcal{D}_t, \alpha_1), \ldots, S_{\rho_r}(\mathcal{D}_t, \alpha_r)) = \ell}.\]
      If $\delta(t) = 1$, observe that for each $i \in [r]$, $S_{\rho_i}(\mathcal{D}_t, \alpha_i) = S_{\rho_i}(\mathcal{D}_{t_1}, \alpha_i^1)$.
      So, we get
      \begin{align*}
        Div(S_{\rho_1}(\mathcal{D}_t, \alpha_1), \ldots, S_{\rho_r}(\mathcal{D}_t, \alpha_r))
        = Div(S_{\rho_1}(\mathcal{D}_{t_1}, \alpha_1^1), \ldots, S_{\rho_r}(\mathcal{D}_{t_1}, \alpha_r^1)) =
        \ell_1 = \ell.
      \end{align*}
      If $\delta(t) = 2$, observe that for each $i \in [r]$,
      $S_{\rho_i}(\mathcal{D}_t, \alpha_i) = S_{\rho_i}(\mathcal{D}_{t_1}, \alpha_i^1) \cup S_{\rho_i}(\mathcal{D}_{t_2}, \alpha_i^2)$,
      where $S_{\rho_1}(\mathcal{D}_{t_1}, \alpha_1^1), \ldots, \allowbreak S_{\rho_r}(\mathcal{D}_{t_1}, \alpha_r^1) \subseteq V(G_{t_1})$,
      $S_{\rho_1}(\mathcal{D}_{t_2}, \alpha_1^2), \ldots, S_{\rho_r}(\mathcal{D}_{t_2}, \alpha_r^2) \subseteq V(G_{t_2})$ and
      ${V(G_{t_1}) \cap V(G_{t_2}) = \emptyset}$. So by \cref{lemma},
      \begin{align*}
        &Div(S_{\rho_1}(\mathcal{D}_t, \alpha_1), \ldots, S_{\rho_r}(\mathcal{D}_t, \alpha_r)) \\
        &= Div(S_{\rho_1}(\mathcal{D}_{t_1}, \alpha_1^1), \ldots, S_{\rho_r}(\mathcal{D}_{t_1}, \alpha_r^1)) +
         Div(S_{\rho_1}(\mathcal{D}_{t_2}, \alpha_1^2), \ldots, S_{\rho_r}(\mathcal{D}_{t_2}, \alpha_r^2)) - |V(G)| f(0^r) \\
        &= \ell_1 + \ell_2 - |V(G)| f(0^r) = \ell.
        \qedhere
      \end{align*}
  \end{proof}

  \medskip

We would like to briefly explain at this point why the factor $d^a$ is stated in the running time of \cref{twThm}.
Integer $a$ stands for the maximum number of children of a node in the given tree decomposition.
The reason why this improvement is possible is that during the DP composition, we need to remember for each entry only the best value of diversity and not all $d$ of them, similarly to our proof of \cref{mainThm}.
Such a change requires only a single computation for each table entry compared to the $d^a$ factor used originally.

\section{Example: The Diverse Vertex Cover Problem}\label{sec:VC}

In the \textsc{Vertex Cover} problem given a graph $G$ and integer $k$, one asks whether there exists $S \subseteq V(G)$
of size at most $k$ such that $S$ touches all the edges of $G$.
Formally, a k-\textsc{Vertex Cover} is a set of pairs $(G, S)$ where $S \subseteq V(G)$
such that $|S| \leq k$ and for each $uv \in E(G)$ $u \in S$ or $v \in S$. %
In the diverse version of this problem, given a graph $G$ and integers $k, r, d$, one asks whether there exist $r$ vertex covers in $G$, each of size at most $k$, such that their diversity is at least $d$. 

We start by showing a dynamic program that, given a cliquewidth decomposition $\mathcal{D}$ of a graph $G$ of width $\omega$,
decides whether there exists a k-\textsc{Vertex Cover} of $G$ in time $f(\omega, k) \cdot |V(\mathcal{D})|$.

\subsection{Dynamic Programming for Vertex Cover}\label{sec:monotoneVC}
Let $\mathcal{D}$ be a cliquewidth decomposition of $G$ of width $\omega$. For $t \in V(\mathcal{D})$, let
$\mathcal{F}_t$ be a set of tuples $(c_1, \ldots, c_\omega)$ such that in $G_t$ there exist a vertex cover $S$ of size at most
$k$ that uses $c_i$ vertices of color $i$ for each $i \in [\omega]$.
Note that $|\mathcal{F}_t| \leq (k + 1)^\omega$. %
Our dynamic programming constructs for each node $t$ set $\mathcal{F}_t$ as follows:
\begin{enumerate}
  \item $t = intro(v, a)$ \\
  Let \lv{$C_{\{v\}} = (c_1, \ldots, c_\omega)$, where 
  $c_i = \begin{cases}
    1 & \text{if $i = a$},\\
    0 & \text{otherwise},
   \end{cases}$} \sv{$C_{\{v\}}$}
  be a tuple that corresponds to a vertex cover $\{v\}$ and \lv{$C_{\emptyset} = (c_1, \ldots, c_\omega)$, where $c_i = 0$ for each $i \in [\omega]$,} \sv{$C_{\emptyset}$}
  be a tuple that corresponds to a vertex cover $\emptyset$. Then
  \begin{equation*}
  \mathcal{F}_t = \{C_{\{v\}}, C_{\emptyset}\}  
  \end{equation*}%
  and $\mathcal{F}_t$ can be computed in time $2\omega$.
  \item $t = t_1 \oplus t_2$ \\
  Observe that a set $S$ is a vertex cover of $G_t$ if and only if it is a union of vertex covers $S_1$ of $G_{t_1}$ and $S_2$ of $G_{t_2}$. \\
  Denote $(x_1, \ldots, x_\omega) + (y_1, \ldots, y_\omega) = (x_1 + y_1, \ldots, x_\omega + y_\omega)$ and $|(x_1, \ldots, x_\omega)| = x_1 + \cdots + x_\omega$. Then
  \begin{equation*}
    \mathcal{F}_t = \{C_1 + C_2 : C_1 \in \mathcal{F}_{t_1}, C_2 \in \mathcal{F}_{t_2}, |C_1 + C_2| \leq k\}  
  \end{equation*}
  and
  $\mathcal{F}_t$ can be computed in time $(k + 1)^{2\omega} \cdot \omega$.
  \item $t = recolor(t', a \rightarrow b)$ \\
  Denote $\rho_{a \rightarrow b}(c_1', \ldots, c_\omega') = (c_1, \ldots, c_\omega)$ where
  $c_i = \begin{cases}
    0 & \text{if $i = a$},\\
    c_a' + c_b' & \text{if $i = b$},\\
    c_i' & \text{otherwise}.
   \end{cases}$ \\
  Then
  \begin{equation*}
    \mathcal{F}_t = \{\rho_{a \rightarrow b}(C') : C' \in \mathcal{F}_{t'}\}
  \end{equation*}
  and
  $\mathcal{F}_t$ can be computed in time $(k + 1)^{\omega} \cdot 2$.
  \item $t = addEdges(t', a, b)$ \\
  Observe that a set $S$ is a vertex cover of $G_t$ if and only if it is a vertex cover of $G_{t'}$ and contains
  all vertices of color $a$ or all vertices of color $b$.
  Let $n_i$ be the number of vertices of color $i$ in $G_t$.
  We can compute it by induction for each $i \in [\omega]$ for each $t \in V(\mathcal{D})$ in time $|V(\mathcal{D})|$.
  Then
\begin{equation*}
    \mathcal{F}_t = \{(c_1, \ldots, c_\omega) \in \mathcal{F}_{t'} | c_a = n_a \lor c_b = n_b\}
\end{equation*}
  and $\mathcal{F}_t$ can be computed in time $(k + 1)^{\omega} \cdot 2$.
\end{enumerate}
\begin{theorem}
  Given a cliquewidth decomposition $\mathcal{D}$ of $G$ of width $\omega$ and an integer $k > 0$ we can determine whether
  $G$ admits a k-\textsc{Vertex Cover} in time
  \begin{equation*}
   \mathcal{O}(|V(\mathcal{D})| \cdot (k + 1)^{2\omega} \cdot \omega).
  \end{equation*}
\end{theorem}

\begin{proof}
  For each node $t \in V(\mathcal{D})$, set $\mathcal{F}_t$ can be computed in time
  \begin{equation*}
    max(2\omega, (k + 1)^{2\omega} \cdot \omega, (k + 1)^{\omega} \cdot 2).
  \end{equation*}
  By the definition of $\mathcal{F}_t$, $G$ has a k-\textsc{Vertex Cover} if and only if
  $\mathcal{F}_q$ is not empty, where $q$ is the root of~$\mathcal{D}$.
\end{proof}

\subsection{From Vertex Cover to Diverse Vertex Cover}

Now we will show how the above approach can be generalized to solve the diverse version of k-\textsc{Vertex Cover} on $G$ in time $f(\omega, k, r) \cdot |V(\mathcal{D})|$, where $r$ is the number of solutions. 
Note that the complexity does not depend on the aimed diversity value $d$.

We define dynamic programming entries representing the tuples of partial solutions
for each node $t$ of the decomposition.
Intuitively, these entries should be combinations of entries from $\mathcal{F}_t$.
For $t \in V(\mathcal{D})$ let $\mathcal{R}_t$ be a set of tuples
\begin{equation*}
  ((c_1^1, \ldots, c_\omega^1), (c_1^2, \ldots, c_\omega^2), \ldots, (c_1^r,\ldots, c_\omega^r), \ell)
\end{equation*}
such that there exist $r$ vertex covers $S_1,\ldots, S_r$ in $G_t$, each of size at most $k$, such that for each $j \in [r], \ i \in [\omega]$
$S_j$ uses $c^j_i$ vertices of color $i$ and
\begin{equation*}
  \ell = max\{Div(S_1, \ldots, S_r) : \forall_{j \in [r]} \forall_{i \in [\omega]} \ S_j \textrm{ is a vertex cover of } G_t \textrm{ that uses } c_i^j \textrm{ vertices of color } i\}.
\end{equation*}
Note that $|\mathcal{R}_t| \leq (k + 1)^{r\omega}$.\\
Our dynamic programming constructs for each node $t$ set $\mathcal{R}_t$ as follows:
\begin{enumerate}
  \item $t = intro(v, a)$
  \begin{equation*}
    \mathcal{R}_t = \{(C^1, \ldots, C^r, \ell) : \forall_{j \in [r]} \ C^j \in \{C_{\{v\}}, C_{\emptyset}\}, \ell = f(m(C^1, \ldots, C^r, v))\},
  \end{equation*}
  where $m(C^1, \ldots, C^r, v) \in \{0, 1\}^r$,
  $m(C^1, \ldots, C^r, v)[j] = \begin{cases}
      1 & \text{if $C^j = C_{\{v\}}$},\\
      0 & \text{otherwise}.
   \end{cases}$ \\
  $\mathcal{R}_t$ can be computed in time $2^r \cdot r\omega$.
  \item $t = t_1 \oplus t_t$
  \begin{multline*}
    \mathcal{R}_t = \{(C^1, \ldots, C^r, \ell) :
    \exists_{(C_1^1, \ldots, C_1^r, \ell_1) \in \mathcal{R}_{t_1}, (C_2^1, \ldots, C_2^r, \ell_2) \in \mathcal{R}_{t_2}} \ s.t. \ \forall_{j \in [r]} \ C^j = C_1^j + C_2^j \land |C^j| \leq k, \\
    \ell = max\{\ell_1 + \ell_2 - |V(G)| f(0^r) : \exists_{(C_1^1, \ldots, C_1^r, \ell_1) \in \mathcal{R}_{t_1}, (C_2^1, \ldots, C_2^r, \ell_2) \in \mathcal{R}_{t_2}} \ s.t. \ \forall_{j \in [r]} C^j = C_1^j + C_2^j \land |C^j| \leq k\}\}.
  \end{multline*}
  $\mathcal{R}_t$ can be computed in time $(k + 1)^{2r\omega} \cdot r\omega$. \\
  Since $V(G_{t_1})$ and $V(G_{t_2})$ are disjoint, $(C_1^1, \ldots, C_1^r, \ell_1)$ represents a tuple $S_1, \ldots S_r \subseteq V(G_{t_1})$ of diversity $\ell_1$ and $(C_2^1, \ldots, C_2^r, \ell_2)$ represents a tuple $P_1, \ldots P_r \subseteq V(G_{t_2})$ of diversity $\ell_2$,
  by \Cref{lemma} the diversity of $S_1 \cup P_1, \ldots, S_r \cup P_r$ represented by $(C^1, \ldots, C^r, \ell)$ is $\ell_1 + \ell_2 - |V(G)| \cdot f(0^r)$. Observe that $(C^1, \ldots, C^r)$ may represent
  multiple tuples of partial solutions in $G_t$, but since we are interested in maximizing the diversity, it suffices to keep only the best diversity value among all these tuples.
  
  \item $t = recolor(t', a \rightarrow b)$
  \begin{equation*}
    \mathcal{R}_t = \{(\rho_{a \rightarrow b}(C^1), \ldots, \rho_{a \rightarrow b}(C^r), \ell)
    : (C^1, \ldots, C^r, \ell) \in \mathcal{R}_{t'}\}.
  \end{equation*}
  $\mathcal{R}_t$ can be computed in time $(k + 1)^{r\omega} \cdot 2r$.
  \item $t = addEdges(t', a, b)$
  \begin{equation*}
    \mathcal{R}_t = \{((c_1^1,\ldots, c_\omega^1), \ldots, (c_1^r, \ldots, c_\omega^r), \ell) \in \mathcal{R}_{t'} : \forall{j \in [r]} \ c_a^j = n_a \lor c_b^j = n_b\}.
  \end{equation*}
  $\mathcal{R}_t$ can be computed in time $(k + 1)^{r\omega} \cdot 2r$.
\end{enumerate}

\begin{theorem}
  Given a cliquewidth decomposition $\mathcal{D}$ of $G$ of width $\omega$ and integers $k, r, d > 0$, we can determine whether
  $G$ admits $r$ k-\textsc{Vertex Covers} of diversity at least $d$ in time
  \begin{equation*}
   \mathcal{O}(|V(\mathcal{D})| \cdot (k + 1)^{2r\omega} \cdot r\omega).
  \end{equation*}
  
\end{theorem}

\begin{proof}
  For each $t \in V(\mathcal{D})$, $\mathcal{R}_t$ can be computed in time
  \begin{equation*}
    max(2^r \cdot r\omega, (k + 1)^{2r\omega} \cdot r\omega, (k + 1)^{r\omega} \cdot 2r).
  \end{equation*}
  By the definition of $\mathcal{R}_t$, $G$ admits $r$ k-\textsc{Vertex Covers} of diversity at least $d$ if and only if
  ${(C^1, \ldots, C^r, \ell) \in \mathcal{R}_q}$ for some $C^1, \ldots, C^r \in \mathbb{N}^\omega$ and $\ell \geq d$, where $q$ is the root of $\mathcal{D}$.
\end{proof}

\subsection{Dynamic Programming Core for \textsc{Vertex Cover}}

We describe dynamic programming for \textsc{Vertex Cover} from \Cref{sec:monotoneVC} in terms of the dynamic programming core and membership function.
For each node $t \in V(\mathcal{D})$, we define $Process_{\mathfrak{C}, \mathcal{D}}(t)$ as follows:
\begin{enumerate}
  \item $t = intro(v, a)$
  \begin{equation*}
  Process_{\mathfrak{C}, \mathcal{D}}(t) = \mathcal{F}_t = \{C_{\{v\}}, C_{\emptyset}\}.
  \end{equation*}
  \item $t = t_1 \oplus t_2$
  \begin{equation*}
    Process_{\mathfrak{C}, \mathcal{D}}(t) = \{(C_1 + C_2, C_1, C_2) : C_1 \in \mathcal{F}_{t_1}, C_2 \in \mathcal{F}_{t_2}, |C_1 + C_2| \leq k\},
  \end{equation*}
  where $C_1, C_2 \in \mathbb{N}^\omega$, $C_1 + C_2$ and $|C_1 + C_2|$ for $\omega$-tuples defined as before.\\
  \begin{equation*}
    Process_{\mathfrak{C}, \mathcal{D}}(t) = \{(\rho_{a \rightarrow b}(C'), C') : C' \in \mathcal{F}_{t'}\}.
  \end{equation*}
  \item $t = addEdges(t', a, b)$
  \begin{equation*}
    Process_{\mathfrak{C}, \mathcal{D}}(t) = \{(C', C') : C' = (c_1, \ldots, c_\omega) \in \mathcal{F}_{t'} , c_a = n_a \lor c_b = n_b\},
  \end{equation*}
  where $n_i$ denotes the number of vertices of color $i$ in $G_t$.
\end{enumerate}

Observe that $Process_{\mathfrak{C}, \mathcal{D}}(t)$ describes exactly how each entry in $\mathcal{F}_t$ is computed from entries in $\mathcal{F}_{t_1}, \ldots, \mathcal{F}_{t_{\delta(t)}}$.
$Process_{\mathfrak{C}, \mathcal{D}}(t)$ is empty if and only if $\mathcal{F}_t$ is empty, so we can define $Accept_{\mathfrak{C}, \mathcal{D}}$ to be $\{0, \ldots, k\}^\omega$.
$Accept_{\mathfrak{C}, \mathcal{D}}$ can be computed in time $(k + 1)^\omega$.
Moreover, for each $t \in V(\mathcal{D})$, $Process_{\mathfrak{C}, \mathcal{D}}(t)$ can be computed in the same time as $\mathcal{F}_t$.
So $\tau(\mathfrak{C}, \mathcal{D}) = \mathcal{O}(|V(\mathcal{D})| \cdot (k + 1)^{2\omega} \cdot \omega)$.
Observe that for each $t \in V(\mathcal{D})$, for each entry ${(C, C_1, \ldots C_{\delta(t)}) \in Process_{\mathfrak{C}, \mathcal{D}}(t)}$, the values of $C_1, \ldots C_{\delta(t)}$ determine the value of $C$ and $\delta(t) \leq 2$, so $Size_{\mathfrak{C}, \mathcal{D}} \leq (k + 1)^{2\omega}$.

It remains to prove that $\mathfrak{C}$ is monotone. Let $\rho(C) = 0$ if $C = C_\emptyset$, $1$ otherwise.
\begin{lemma}\label{lem:monotoneVC}
  For each graph $G$, for each decomposition $\mathcal{D}$ of $G$, for each node $t \in V(\mathcal{D})$
  there exists a vertex cover $S \subseteq V(G_t)$ of $G_t$ if and only if there exists
  a $(\mathfrak{C}, \mathcal{D}_t, w)$-witness $\alpha$ such that $S = S_\rho(\mathcal{D}_t,\alpha)$.
  Moreover, in this case, $w = (c_1, \ldots, c_\omega)$, where $c_i$ is the number of vertices of color $i$ in $S$.
\end{lemma}
\begin{proof}
  We use induction on the decomposition tree.
  \paragraph{\textbf{Node $\mathbf{t = Leaf_{\mathcal{D}}(v)}$.}} 
    Let $a \in$ be the color of $v$. There are two possible witnesses: $\alpha_1(t) = C_{\{v\}}$ and $\alpha_2(t) = C_\emptyset$ and two possible vertex covers of
    $G_t$: $\{v\}$ and $\emptyset$. Observe that $S_\rho(\mathcal{D}_t, \alpha_1) = \{v\}$ and $\{v\}$ uses one vertex of color $a$,
    $S_\rho(\mathcal{D}_t, \alpha_2) = \emptyset$ and $\emptyset$ uses no vertices of any color.
  \paragraph{\textbf{Node $\mathbf{t}$ has children $\mathbf{t_1, \ldots, t_{\delta(t)}}$.}}
    Let $\alpha$ be a $(\mathfrak{C}, \mathcal{D}_t, w)$-witness and $S \coloneq S_\rho(\mathcal{D}_t, \alpha)$. We will prove that $w = (c_1, \ldots, c_\omega)$ and $S$ is a vertex cover of $G$ that uses $c_i$ vertices of color $i$.
    Let $w_j \coloneq \alpha(t_j)$ and $S_j \coloneq S_\rho(\mathcal{D}_{t_j}, \alpha_j)$ for $j \in [\delta(t)]$.
    Observe that $S = \bigcup_{j \in [\delta(t)]} S_j$ and
    $\alpha_j \coloneq \alpha |_{V(\mathcal{D}_{t_j})}$
    is a $(\mathfrak{C}, \mathcal{D}_{t_j}, w_j)$-witnesses.
    Then by the induction assumption, $S_j$ is a vertex cover of $G_{t_j}$ and $w_j = (c_1^j, \ldots, c_\omega^j)$ where
    $c_i^j$ is the number of vertices of color $i$ in $S_j$.
    If $t = t_1 \oplus t_2$ or $t = recolor(t_1, a \rightarrow b)$, $E(G_t) = \bigcup_{j \in [\delta(t)]} E(G_{t_j})$, so
    $S$ is a vertex cover of $G_t$. Moreover, $S$ uses $\sum_{j \in [\delta(t)]} c_i^j = c_i$ vertices of color $i$ for each $i \in [\omega]$.
    If $t = addEdges(t_1, a, b)$ and $(w, w_1) \in Process_{\mathfrak{C}, \mathcal{D}}(t)$, then $c_a^1 = n_a$ or $c_b^1 = n_b$, so all of
    $E(G_t) \setminus E(G_{t_1})$ are covered by $S_1$. Therefore, $S = S_1$ is a vertex cover of $G_t$.
    Moreover, $S$ uses $c_i^1 = c_i$ vertices of color $i$ for each $i \in [\omega]$.

    Now let $S$ be a vertex cover of $G_t$ that uses $c_i$ vertices of color $i$ for each $i \in [\omega]$. We will prove that there exists a $(\mathfrak{C}, \mathcal{D}_t, w)$-witness $\alpha$ for $w = (c_1, \ldots, c_\omega)$ such that $S = S_\rho(\mathcal{D}_t, \alpha)$. Let $S_j \coloneq S \cap V(G_{t_j})$ for each $j \in [\delta(t)]$.
    Observe that $S_j$ is a vertex cover of $G_{t_j}$. Then, by the induction assumption, there exist a $(\mathfrak{C}, \mathcal{D}_{t_j}, w_j)$-witness $\alpha_j$
    such that $S_j = S_\rho(\mathcal{D}_{t_j}, \alpha_j)$ and $w_j = (c_1^j, \ldots, c_\omega^j)$, where $c_i^j$ is the number of vertices of color $i$ in $S_j$.
    Observe that if there exists a $(\mathfrak{C}, \mathcal{D}_{t_j}, w_j)$-witness, then $w_j \in \mathcal{F}_{t_j}$.
    So, by the definition of $Process_{\mathfrak{C}, \mathcal{D}}(t)$ there exists $(w, w_1, \ldots, w_{\delta(t)}) \in Process_{\mathfrak{C}, \mathcal{D}}(t)$
    such that $w = (c_1, \ldots, c_\omega)$, where $c_i = \sum_{j \in [\delta(t)]}c_i^j$ for each $i \in [\omega]$.
    Let us define for each $t' \in V(\mathcal{D}_t)$,
    \[
      \alpha(t') \coloneq \begin{cases}
        \alpha_j(t')  & \text{if $t' \in V(\mathcal{D}_{t_j}), \ j \in \delta(t)$,} \\
        w & \text{if $t' = t$.}
      \end{cases}  
    \]
    Observe that $\alpha$ is a $(\mathfrak{C}, \mathcal{D}, w)$-witness and $S_\rho(\alpha, \mathcal{D}_t) = \bigcup_{j \in [\delta(t)]} S_j = S$, so $S$ uses $c_i$ vertices of color $i$
    for each $i \in [\omega]$.
\end{proof}

The direct conclusion of \cref{theorem} is the following corollary.

\begin{corollary}
  Given a cliquewidth decomposition $\mathcal{D}$ of $G$ of width $\omega$ and integers $k, r, d > 0$, we can determine whether
  $G$ admits $r$ k-\textsc{Vertex Covers} of diversity at least $d$ in time 
 
  \begin{equation*}
    \mathcal{O}(|V(\mathcal{D})| \cdot ((k + 1)^{2\omega} \cdot r\omega + (k + 1)^{2r\omega})).
  \end{equation*}
\end{corollary}

\subsection{Example of Non-monotone DP}\label{sec:nonMonotone}

Consider another variant of the $k$-\textsc{Vertex Cover} problem called $k$-\textsc{Min Vertex Cover}, where
given a graph $G$ one asks to find a vertex cover of $G$ of minimal size, but at most $k$.
Formally, $k$-\textsc{Min Vertex Cover} is set of pairs $(G, S)$ where $S \subseteq V(G)$ such that $|S| \leq k$, for each $uv \in E(G)$ $u \in S$ or $v \in S$
and for each $S' \subseteq V(G)$ such that $|S'| < |S|$ there exists $uv \in E(G)$ such that $u \not \in S'$ and $v \not \in S'$.

Let $\mathcal{D}$ be a cliquewidth decomposition of $G$ of width $\omega$. For $t \in V(\mathcal{D})$, we define
set $\mathcal{F}_t$ as a set of tuples $(c_1, \ldots, c_\omega)$ such that in $G_t$ there exist an inclusion-wise minimal vertex cover $S$
of size at most $k$
that uses $c_i$ vertices of color $i$ for each $i \in [\omega]$.
Note that $|\mathcal{F}(t)| \leq (k + 1)^\omega$.
Any vertex cover of minimal size is also inclusion-wise minimal, so if we could find all the inclusion-wise minimal vertex covers,
these of the minimal size among them would be the solution. \\
Our dynamic programming constructs for each node $t$ set $\mathcal{F}_t$ as follows:
\begin{enumerate}
  \item $t = intro(v, a)$ \\
  $G_t$ is edgeless, so the only inclusion-wise minimal vertex cover is an empty set.
  Let $C_{\emptyset}$
  be a tuple that corresponds to a vertex cover $\emptyset$.
  Then
  \begin{equation*}
  \mathcal{F}_t = \{C_{\emptyset}\}.
  \end{equation*}
  \item $t = t_1 \oplus t_2$ \\
  Observe that a set $S$ is an inclusion-wise minimal vertex cover of $G_t$ if and only if it is a union of inclusion-wise minimal vertex covers $S_1$ of $G_{t_1}$ and $S_2$ of $G_{t_2}$.
  Recall that for $C_1 = (x_1, \ldots, x_\omega), \ C_2 = (y_1, \ldots, y_\omega)$ we denote by $C_1 + C_2$ the tuple ${(x_1 + y_1, \ldots, x_\omega + y_\omega)}$.
  Then
  \begin{equation*}
    \mathcal{F}_t = \{C_1 + C_2 : C_1 \in \mathcal{F}_{t_1}, C_2 \in \mathcal{F}_{t_2}, |C_1 + C_2| \leq k\}.
  \end{equation*}
  \item $t = recolor(t', a \rightarrow b)$
  \begin{equation*}
    \mathcal{F}_t = \{\rho_{a \rightarrow b}(C') : C' \in \mathcal{F}_t'\},
  \end{equation*}
  where $\rho_{a \rightarrow b}(C')$ defined as previously.
  \item $t = addEdges(t', a, b)$ \\
  Observe that a set $S$ is an inclusion-wise minimal vertex cover in $G_t$ if and only if it is a union of an inclusion-wise minimal vertex cover $S'$ in $G_{t'}$ and
  an inclusion-wise minimal set of vertices that covers new edges, which is either the set of all vertices of color $a$ or the set of all vertices of color $b$.
  Recall that $n_i$ denotes the number of vertices of color $i$ in $G_t$, and
  we can compute it by induction for each $i \in [\omega]$ for each $t \in V(\mathcal{D})$ in time $|V(\mathcal{D})|$.
  \\
  For $a \in [\omega]$ let $all_a(c_1', \ldots, c_\omega') \coloneq (c_1, \ldots c_\omega)$ where
  $c_i = \begin{cases}
    n_a & \text{if $i = a$},\\
    c_i' & \text{otherwise}.
   \end{cases}$\\
  Then
  \begin{equation*}
    \mathcal{F}_t = \{all_a(C') : C' \in \mathcal{F}_{t'}\} \cup \{all_b(C') : C' \in \mathcal{F}_{t'}\}.
  \end{equation*}
\end{enumerate}

The core corresponding to this dynamic programming is clearly not monotone. The value of any $(\mathfrak{C}, \mathcal{D})$-witness $\alpha$
on the leaves of $\mathcal{D}$ is equal to $C_{\emptyset}$, so for any membership function $\rho$ either $S_\rho(\mathcal{D}, \alpha) = \emptyset$ or
$S_\rho(\mathcal{D}, \alpha) = V(G)$. \\
However, having the witness $\alpha$, we can still deduce the corresponding solution set by looking at the value of $\alpha$ on the nodes of the last type, $t = addEdges(t', a, b)$,
where either the whole $a$-part or the whole $b$-part of $G_t$ is included to the solution.

Observe that we can also solve $k$-\textsc{Min Vertex Cover} by running the same dynamic programming as for $k$-\textsc{Vertex Cover}
and then choosing the solution of the smallest size among the found ones.
The worst-case running time is asymptotically the same, so we do not benefit from dropping the monotonicity constraint.

\section{\texorpdfstring{MSO$_1$}{MSO1} Has Monotone DP}\label{sec:MSO}

\begin{theorem}[Courcelle, Makowsky, and Rotics~\cite{CMR}]\label{Croucelle}
  Given a cliquewidth decomposition $\mathcal{D}$ of a graph $G$, an $MSO_1$ formula $\Phi$ can be evaluated on $G$ in time $c\cdot  |V(\mathcal{D})|$, where $c$ is a constant depending only on $\Phi$ and the width of~$\mathcal{D}$. 
\end{theorem}
The proof of \cref{Croucelle} will give us dynamic programming over the decomposition of $G$.
In \cref{sec:MSOMonotone}, we will show that for vertex problems, this dynamic programming is monotone, which, combined with \cref{theorem}, gives us \cref{cor:mainCor}.

\subsection{Proof of Courcelle's Theorem for Cliquewidth Using Reduced Evaluation Trees}\label{sec:CroucelleProof}

\sv{
A \textit{partial evaluation tree} $F_t$ for $\Phi$ and a node $t$ of the decomposition $\mathcal{D}$ of $G$
is the evaluation tree for $\Phi$ and $G_t$ in which
for individual variables, we consider the additional case when the value does not
belong to $G_t$ (we call such branch \textit{external}).
Observe that if we prune from $F_t$ all the external branches, we get the evaluation tree for $\Phi$ and $G_t$.
We can reduce partial evaluation trees similarly as regular evaluation trees,
defining the isomorphism classes of the leaves (also called \textit{configurations}) as combinations of: the colored subgraph $G_t'$ of $G_t$ induced by the values of individual variables corresponding to this leaf,
assignment of the individual variables and assignment of the set variables restricted to $V(G_t')$
(there are at most $2^{\mathcal{O}(q_v(q + \log \omega))}$ such combinations, where $\omega$ is the number of colors).
Proceeding with the same bottom-up reduction, we get a \textit{reduced partial evaluation tree} $F_t^-$
of size bounded by
\lv{\[
\left.\kern-\nulldelimiterspace
\begin{array}{@{}c@{}}
  2^{\cdot^{\cdot^{\cdot^{\scriptstyle 2^{\mathcal{O}(q_v(q + \log \omega))}}}}}
\end{array}
\right\rbrace
\text{\scriptsize $q + 1$}
\]}
\sv{$
\left.\kern-\nulldelimiterspace
\begin{array}{@{}c@{}}
  2^{\cdot^{\cdot^{\cdot^{\scriptstyle 2^{\mathcal{O}(q_v(q + \log \omega))}}}}}
\end{array}
\right\rbrace
\text{\scriptsize $q + 1$}
$}
which is isomorphic to $F_t$, so after pruning the external branches, it evaluates to the same value as the regular evaluation tree for $G_t$.
It remains to prove that we can compute $F_t^-$ for $t$ being the root of $\mathcal{D}$ without computing $F_t$.
The idea is to construct $F_t^-$ inductively from $F_{t_1}^-, \ldots, F_{t_\delta(t)}^-$.
Observe that for $t$ being the leaf, the size of $F_t$ is bounded, so we can compute $F_t$ and reduce it.
For $t = recolor(t_1, a \rightarrow b)$ or $t = addEdges(t_1, a, b)$, it suffices to reduce $F_{t_1}^-$, taking into account the new colored graph.
For $t = t_1 \oplus t_2$ we compute a \textit{tree product} $F_{t_1}^- \otimes F_{t_2}^-$ and then reduce it.
The tree product of two partial evaluation trees $F_1, F_2$ for the same formula $\Phi$ and graphs $G_1, G_2$ is recursively defined as follows.
If $F_1, F_2$ are leaves, $F_1 \otimes F_2 = (F_1, F_2)$.
If the first quantified variable of $\Phi$ is a set variable, $F_1$ has children $\{F_1^i: i \in [n]\}$ corresponding, respectively, to the assignments of sets $\{S_i: i \in [n]\}$ to this variable,
and similarly $F_2$ has children $\{F_2^j : j \in [m]\}$ corresponding, respectively, to the assignments of $\{P_j: j \in [m]\}$,
then $F_1 \otimes F_2$ has children $\{F_1^i \otimes F_2^j: i \in [n], j \in [m]\}$ corresponding to the assignments of $\{S_i \cup P_j: i \in [n], j \in [m]\}$.
If the first quantified variable of $\Phi$ is an individual variable, $F_1$ has children $F_1^0$ and $\{F_1^i: i \in [n]\}$ corresponding, respectively, to the assignments of external value and $\{v_i: i \in [n]\}$ to this variable,
and similarly $F_2$ has children $F_2^0$ and $\{F_2^j : j \in [m]\}$ corresponding, respectively, to the assignments of external value and $\{u_j : j \in [m] \}$,
then $F_1 \otimes F_2$ has children $F_1^0 \otimes F_2^0$, $\{F_1^i \otimes F_2^0: i \in [n]\}$ and $\{F_1^0 \otimes F_2^j : j \in [m]\}$
respectively, to the assignments of external value, $\{v_i : i \in [n]\}$ and $\{u_j : j \in [m]\}$ to this variable.
Note that $|F_1 \otimes F_2| \leq |F_1| \cdot |F_2|$. It is easy to see that for $t = t_1 \oplus t_2$, $F_t = F_{t_1} \otimes F_{t_2}$.
Then, in order to show the correctness of the inductive construction of reduced partial evaluation trees, it remains to prove the following lemma, whose proof is \cite{OurArxiv}.
}

We prove \cref{Croucelle} using partial evaluation trees.
Let $\Phi$ be an $MSO_1$ formula in the prenex normal form and let $\Psi(z_1, \ldots, z_q)$ be the quantifier-free part of $\Phi$.
Let $q_v$ be the number of variables quantified over single vertices (called \emph{individual variables}) and $q_s$ be the number of variables quantified over sets of vertices (called \emph{set variables}). The total quantifier depth of $\Phi$ is $q = q_v + q_s$

\begin{definition}[Evaluation tree]
  \emph{Evaluation tree} for a formula $\Phi$ and graph $G$ is a rooted tree $F$ with $q + 1$ levels (where 0 is the level of the root) constructed as follows.
  For each level $\ell \in [q]$, each node $w$ of the $\ell$-th level corresponds to a particular choice of the first $\ell$
  variables. If $z_{\ell + 1}$ is an individual variable, then $w$ has $n$ children, representing all possible choices of vertices
  for $z_{\ell + 1}$. If $z_{\ell + 1}$ is a set variable, then $w$ has $2^n$ children, representing all possible choices of subsets of vertices for $z_{\ell + 1}$.
\end{definition}

The answer to whether $\Phi$ is satisfied on $G$ can be obtained by traversing $F$ bottom-up.
The value of a leaf is true if the valuation corresponding to this leaf satisfies $\Phi$.
The value of each internal node depends on the values of its children.
If a node $w$ corresponds to a universally quantified variable, it must have all children satisfied to be satisfied, if it corresponds to an existentially quantified variable, it must have at least one satisfied child to be satisfied. However, $F$ can be large.
But it turns out that its size can be reduced.

\begin{observation}
When a set assigned to a variable $z_i$ contains a vertex $u$ that is not assigned to any individual variable, then $u$ does not participate in any of the membership predicates. Hence, the evaluation of $\Psi$ is
independent of its presence in $z_i$.
\end{observation}

The reduction goes inductively bottom-up. Each leaf of $F$ corresponds to an evaluation of $\Psi$ on
a subgraph $G'$ of $G$ induced by $q_v$ vertices assigned to individual variables.
The number of non-isomorphic graphs on $q_v$ vertices is bounded by $2^{\binom{q_v}{2}}$. The number of possible assignments of vertices of such a graph to individual variables is bounded by ${q_v}^{q_v}$, and the number of possible assignments of sets of vertices to set variables is bounded by $(2^{q_v})^{q_s}$.
In total, we get at most $2^{\mathcal{O}(q_vq)}$ non-isomorphic combinations of $G'$ and variable assignments.
These combinations define isomorphism classes of the leaves. Since leaves from the same isomorphism class have the same
value in the evaluation, the evaluation of node $w$ of level $q$ is independent of whether there are one or more leaves from
the same class among its children, so it suffices to keep only one leaf for each isomorphism class
that appears there. After this reduction, each node of level $q$ has at most $2^{\mathcal{O}(q_vq)}$ children. %
Now we can distinguish nodes of level $q - 1$ by sets of their children. There are at most $2^{2^{\mathcal{O}(q_vq)}}$ possible
sets of children. Each of these sets defines the isomorphism class of trees of height one. Now reduce isomorphic subtrees under
the same parent analogously as leaves. %
This reduction procedure can be repeated consecutively for levels $q - 2, \ldots, 1$.

\begin{definition}[Reduced evaluation tree]
  $F^{-}$ obtained from $F$ by the above reduction is called \emph{reduced evaluation tree} for $\Phi$ and $G$.
\end{definition}

\begin{observation}
  $F^{-}$ can be evaluated instead of $F$ to determine whether $G \vDash \Phi$.
\end{observation}

\begin{observation}
  The size of $F^{-}$ is bounded by
  \begin{equation*}
    \left.\kern-\nulldelimiterspace
    \begin{array}{@{}c@{}}
      2^{\cdot^{\cdot^{\cdot^{\scriptstyle 2^{\mathcal{O}(q_v q)}}}}}
    \end{array}
    \right\rbrace
    \text{\scriptsize $q + 1$.}
  \end{equation*}
\end{observation}

Now we will show that given a cliquewidth decomposition $\mathcal{D}$ of $G$ of width $\omega$, we can construct $F^{-}$ without constructing $F$.

\begin{definition}[Partial evaluation tree]
  \emph{Partial evaluation tree} for a formula $\Phi$ and a node $t$ of the decomposition $\mathcal{D}$ of $G$ is an evaluation tree
  for $\Phi$ and $G_t$ where for each individual variable we additionally consider a choice of external value (not belonging to $G_t$) for this variable,
  and add an extra branch labeled by $ext$ representing this choice.
  That means for each $\ell$ such that $z_{\ell + 1}$ is an individual variable each node of level $\ell$ has $|V(G_t)| + 1$ children.
\end{definition}

\begin{definition}[Tree product]
  Let $\Phi(z_1, \ldots, z_\ell)$ be an $MSO_1$ formula in a prenex normal form, where $z_1, \ldots, z_\ell$ are free variables and
  $z_{\ell + 1}, \ldots, z_q$ are quantified variables. Let $F_1, F_2$ be partial evaluation trees for $\Phi(z_1, \ldots, z_\ell)$
  and $G_{t_1}$ and $G_{t_2}$ respectively,
  $u$ is the root of $F_1$ and $w$ is the root of $F_2$. \emph{Product} of $F_1, F_2$ is a partial evaluation tree $F \coloneq F_1 \otimes F_2$ defined recursively.
  \begin{itemize}
    \item If the height of $F_1, F_2$ is $0$ (that means $\ell = q$), then $F$ is a tree of height $0$ with root $(u, w)$.
    \item If the height of $F_1, F_2$ greater than $0$
    \begin{itemize}
      \item If $z_{\ell + 1}$ is a set variable, $F_1$ has children subtrees $F_1^1, \ldots, F_1^n$ labeled by sets
      $S_1, \ldots, S_n \subseteq V(G_{t_1})$, $F_2$ has children subtrees $F_2^1, \ldots, F_2^m$ labeled by sets
      $P_1, \ldots, P_m \subseteq V(G_{t_2})$, then $F$ is a tree with root $(u, v)$ which for each $i \in [n], j \in [m]$ has a child subtree $F_1^i \otimes F_2^j$ labeled by $S_i \cup P_j$.
      Observe that $F_1^i$ and $F_2^j$ are partial evaluation trees for $\Phi(z_1, \ldots, z_\ell, z_{\ell + 1})$ of height one smaller that $F_1, F_2$, so their product is well-defined.
      \item If $z_{\ell + 1}$ is an individual variable, $F_1$ has children subtrees $F_1^0, F_1^1, \ldots, F_1^n$ labeled by
      $ext, s_1, \ldots, s_n \in V(G_{t_1})$, $F_2$ has children subtrees $F_2^0, F_2^1, \ldots, F_2^m$ labeled by
      $ext, p_1, \ldots, p_m \in V(G_{t_2})$, then $F$ is a tree with root $(u, v)$ which has:
      for each $i \in [n]$ a child subtree $F_1^i \otimes F_2^0$ labeled by $s_i$, for each $j \in [m]$ a child subtree $F_1^0 \otimes F_2^j$ labeled by $p_j$
      and a child subtree $F_1^0 \otimes F_2^0$ labeled by $ext$.
    \end{itemize}
  \end{itemize}
\end{definition}

Observe that if $t = t_1 \oplus t_2$, then $F_{t_1} \otimes F_{t_2} = F_t$.
If $t = recolor(t', a \rightarrow b)$ or $t = addEdges(t', a, b)$, $F_t = F_{t'}$.
Therefore, for each $t \in V(\mathcal{D})$ we can construct a partial evaluation tree $F_t$ for $\Phi$ and $G_t$ by induction on the decomposition tree.
However, the size of such partial evaluation trees can be large.

We will show that we can reduce partial evaluation trees bottom up, similarly to regular evaluation trees.
By being able to perform this reduction after each step of the construction, we will keep the size of partial evaluation trees for the inner nodes of the decomposition small and make the construction efficient.

We start by defining the isomorphism classes of the leaves.
\begin{definition} [Configuration]
Let $\mathcal{D}$ be a cliquewidth decomposition of width $\omega$, $F_t$ be a partial evaluation tree for a formula $\Phi$, and $t \in V(\mathcal{D})$.
A~\emph{configuration} of a node $u$ of $F_t$ of level $\ell$ is a combination of:
\begin{itemize}
  \item a subgraph $G_t'$ of $G_t$ induced by at most $\ell$ vertices -- values of the individual variables among the first $\ell$ variables,
  \item a coloring of these vertices $\mu: V(G_t') \rightarrow [\omega]$,
  \item an assignment $\rho_I$ of values from $V(G_t') \cup \{ext\}$ to the individual variables among the first $\ell$ variables,
  \item an assignment $\rho_S$ of subsets of $V(G_t')$ to set variables among the first $\ell$ variables.
\end{itemize}  
\end{definition}
For the leaves of $F_t$ (nodes of level $q$) there are at most $2^{\binom{q_v}{2}} \cdot \omega^{q_v} \cdot (q_v + 1) ^ {q_v} \cdot (2^{q_v})^{q_s} = 2^{\mathcal{O}(q_v(q + \log \omega))}$
different configurations. These configurations define isomorphism classes of the leaves (subtrees of height $0$ in $F_t$).
For each $h \in [q]$, the isomorphism classes of the subtrees of height $h$ are defined by the sets of isomorphism classes of their children.

Performing the same reduction as for the regular evaluation tree, we get \emph{reduced partial evaluation tree}~$F_t^{-}$.
\begin{observation} \label{reduced-tree-size}
The size of $F_t^{-}$ is bounded by
\begin{align*} 
  \left.\kern-\nulldelimiterspace
  \begin{array}{@{}c@{}}
    2^{\cdot^{\cdot^{\cdot^{\scriptstyle 2^{\mathcal{O}(q_v(q + \log \omega))}}}}}
  \end{array}
  \right\rbrace
  \text{\scriptsize $q + 1$.}
\end{align*}  
\end{observation}

\begin{observation}
  If we prune in $F_t^{-}$ all branches containing external choices, we get an evaluation tree isomorphic to the full evaluation tree for $\Phi$ and $G_t$.
\end{observation}

In order to make the reduction procedure deterministic, for each isomorphism class of trees of height $h$, we can fix one representative.
Then $reduce(F)$ replaces $F$ by the representative of its isomorphism class.

It remains to be proven that we can perform this reduction after each step of the construction of a partial evaluation tree for $G$.

\begin{restatable}{lemma}{lemCombiningCor} \label{combining-cor}
  Let $t_1, t_2 \in V(\mathcal{D})$ such that $G_{t_1}, G_{t_2}$ are disjoint, $F_1 \sim F_1'$ be partial evaluation trees for $\Phi$ and $t_1$,
  $F_2 \sim F_2' $ be partial evaluation trees for $\Phi$ and $t_2$.
  Then $F_1 \otimes F_2 \sim F_1' \otimes F_2'$.
\end{restatable}

We start by a few observations.
\begin{observation}
  Let $F_1, F_2$ be partial evaluation trees for $\Phi$,
  $u, v$ be nodes of level $\ell > 0$ of $F_1$ and $F_2$, respectively.
  If the configurations of $u$ and $v$ are the same, then the configurations of their parents are the same.
\end{observation}

\begin{observation} \label{root-obs}
  Let $F_1, F_2$ be partial evaluation trees for $\Phi$,
  $T_1, T_2$ be subtrees of $F_1$ and $F_2$ rooted at the same level.
  Then if $T_1 \sim T_2$, the configurations of their roots are the same.
\end{observation}

We say two configurations of nodes of level $\ell$ are \emph{compatible} if for each $i \leq \ell$,
where $z_i$ is individual variable, the value assigned to $z_i$ is external in at least one of these configurations.
Observe that while computing the tree product $F_1 \otimes F_2$ recursively, we combine exactly these subtrees of $F_1, F_2$ rooted at the same level whose root configurations are compatible.

\begin{definition}[Union of configurations]
  Let $(G_1, \mu_1, \rho_I^1, \rho_S^1), (G_2, \mu_2, \rho_I^2, \rho_S^2)$
  be a pair of compatible configurations of nodes of level $\ell$ of partial evaluation trees for the same formula $\Phi$.
  Their \emph{union} is a configuration $(G, \mu, \rho_I, \rho_S)$, where:
  \begin{itemize}
    \item $G = G_1 \uplus G_2$, i.e. $G$ is a disjoint union of $G_1$ and $G_2$,
    \item $\mu(v) = \begin{cases}
      \mu_1(v) & \text{ if } v \in V(G_1), \\
      \mu_2(v) & \text{ otherwise,}
    \end{cases}$
    \item $\rho_I = \begin{cases}
      \rho_I^1(z_i) & \text{ if } \rho_I^1(z_i) \neq ext, \\
      \rho_I^2(z_i) & \text{ otherwise,}
    \end{cases}$
    \item $\rho_S(z_i) = \rho_S^1(z_i) \uplus \rho_S^2(z_i)$.
  \end{itemize}
\end{definition}

\begin{observation}\label{product-obs}
  Let 
  $F_1, F_2$ be partial evaluation trees for $\Phi$ and $t_1, t_2$, respectively
  and $T_1, T_2$ be subtrees of $F_1, F_2$ rooted at the same level.
  Then $T_1 \otimes T_2$ is a subtree of $F_1 \otimes F_2$ if and only if the
  root configurations of $T_1$ and $T_2$ are compatible.
  Moreover, if $G_{t_1}$ and $G_{t_2}$ are disjoint, the root configuration of $T_1 \otimes T_2$ in
  $F_1 \otimes F_2$ is a union of the root configurations of $T_1$ and $T_2$.
\end{observation}

\cref{combining-cor} is a direct application of the following lemma, with $\ell = 0$.

\begin{lemma}[Combining lemma] \label{combining-lemma}
  Let 
  $t_1, t_2 \in V(\mathcal{D})$ such that $G_{t_1}, G_{t_2}$ are disjoint, $F_1, F_1'$ be partial evaluation trees for $\Phi$ and $t_1$,
  $F_2, F_2' $ be partial evaluation trees for $\Phi$ and $t_2$, $T_1 \sim T_1'$ be subtrees of $F_1, F_1'$, respectively, rooted at level $\ell$,
  $T_2 \sim T_2'$ be subtrees of $F_2, F_2'$, respectively, rooted at level $\ell$.
  Then, if the root configurations of $T_1$ and $T_2$ are compatible, $T_1 \otimes T_2$ is a subtree of $F_1 \otimes F_2$, $T_1' \otimes T_2'$ is a subtree of $F_1' \otimes F_2'$ 
  and moreover ${T_1 \otimes T_2 \sim T_1' \otimes T_2'}$.

\end{lemma}

\begin{proof}
  Assume the root configurations of $T_1$ and $T_2$ are compatible.
  By \cref{root-obs}, the root configurations of $T_1$ and $T_1'$ are the same, similarly the root configurations of $T_2$ and $T_2'$,
  so the root configurations of $T_1'$ and $T_2'$ are also compatible.
  Then, by \cref{product-obs} $T_1 \otimes T_2$ is a subtree of $F_1 \otimes F_2$ and $T_1' \otimes T_2'$ is a subtree of $F_1' \otimes F_2'$.
  
  It remains to prove that ${T_1 \otimes T_2 \sim T_1' \otimes T_2'}$. We do it by induction on the height of the subtrees.
  As $G_{t_1}$ and $G_{t_2}$ are disjoint, by \cref{product-obs} the root configurations of $T_1 \otimes T_2$ and $T_1' \otimes T_2'$
  are the same and equal to the union of the root configuration of $T_1$ and $T_2$.
  For subtrees of height $0$, that means they are isomorphic.
  Assume now that the height is greater than $0$.
  It suffices to prove that for each child subtree of $T_1 \otimes T_2$ in $F_1 \otimes F_2$, there exists an isomorphic child subtree of $T_1' \otimes T_2'$ in $F_1' \otimes F_2'$.
  Proof of the converse implication will be symmetric.
  For each pair of children subtrees $T_1^i, T_2^j$ of $T_1$ and $T_2$, respectively,
  there exists a pair of children subtrees $T_1'^i, T_2'^j$ of $T_1'$ and $T_2'$, respectively,
  such that $T_1'^i \sim T_1^i$ and $T_2'^j \sim T_2^j$.
  If $T_1^i \otimes T_2^j$ is a child subtree of $T_1 \otimes T_2$ in $F_1 \otimes F_2$, the configurations of $T_1^i$ and $T_2^j$ are compatible.
  Then, by \cref{root-obs} the configurations of $T_1'^i$ and $T_2'^j$ are also compatible, so $T_1'^i \otimes T_2'^j$ is a child subtree of $T_1' \otimes T_2'$ in $F_1' \otimes F_2'$.
  By the induction assumption $T_1'^i \otimes T_2'^j \sim T_1^i \otimes T_2^j$ as subtrees of $F_1 \otimes F_2$ and $F_1' \otimes F_2'$ of height one smaller than $T_1 \otimes T_2, T_1' \otimes T_2'$.
  Therefore, ${T_1 \otimes T_2 \sim T_1' \otimes T_2'}$.
\end{proof}

In other words, \cref{combining-cor} says the product of two reduced partial evaluation trees is isomorphic to the product of unreduced ones.
Having this, we are ready to prove the main theorem.

\begin{proof}[Proof of \cref{Croucelle}]
For each $t \in V(\mathcal{D})$, we can construct a reduced partial evaluation tree $F_t^-$ for $\Phi$ which is isomorphic to $F_t$ by induction as follows:
\begin{enumerate}
  \item $t = intro(v, a)$ \\
  The size of $F_t$ is bounded by $2^{\mathcal{O}(q)}$, so we can construct $F_t$ and reduce it. Then
  \begin{equation*}
    F_t^{-} \coloneq reduce(F_t).
  \end{equation*}
  \item $t = t_1 \oplus t_2$
  \\
  Let
  \begin{equation*}
    F_t' \coloneq F_{t_1}^{-} \otimes F_{t_2}^{-}.
  \end{equation*}
  Note that, by \cref{reduced-tree-size},
  \begin{align*}
    |F_t'| \sv{&}\leq |F_{t_1}^{-}| \cdot |F_{t_2}^{-}| \leq
    \left(
    \left.\kern-\nulldelimiterspace
    \begin{array}{@{}c@{}}
      2^{\cdot^{\cdot^{\cdot^{\scriptstyle 2^{\mathcal{O}(q_v(q + \log \omega))}}}}}
    \end{array}
    \right\rbrace
    \text{\scriptsize $q + 1$}
    \right)^2 \sv{\\&} =
    \left.\kern-\nulldelimiterspace
    \begin{array}{@{}c@{}}
      2^{\cdot^{\cdot^{\cdot^{\scriptstyle 2^{\mathcal{O}(q_v(q + \log \omega))}}}}}
    \end{array}
    \right\rbrace
    \text{\scriptsize $q + 1.$}
  \end{align*}
  Some configurations in $F_t'$ may be isomorphic, so we can reduce them. Then
  \begin{equation*}
    F_t^{-} \coloneq reduce(F_t').
  \end{equation*}
  \item $t = recolor(t', a \rightarrow b)$ \\
  If two configurations of leaves of $F^-_{t'}$ were isomorphic, they stay isomorphic after recoloring. However, some non-isomorphic configurations
  may become isomorphic, so we have to reduce them by taking into account the new coloring $\mu$. Then
  \begin{equation*}
    F_t^{-} \coloneq reduce(F_{t'}^{-}).
  \end{equation*}
  \item $t = addEdges(t', a, b)$ \\
  If two configurations of leaves of $F^-_{t'}$ were isomorphic, that means that the induced subgraphs of these configurations were isomorphic,
  including the coloring of vertices, so after adding the new edges, they stay isomorphic. However, some non-isomorphic configurations
  may become isomorphic, so we have to reduce them, taking into account the new edges. Then
  \begin{equation*}
    F_t^{-} \coloneq reduce(F_{t'}^{-}).
  \end{equation*}
\end{enumerate}
We can reduce a partial evaluation tree $F_t^-$ in time $\mathcal{O}(|F_t^{-}|)$.
So the time necessary to build and evaluate $F_t^{-}$ for $t$ being the root of the decomposition is bounded by
\begin{equation*}
  \left.\kern-\nulldelimiterspace
  \begin{array}{@{}c@{}}
    |V(\mathcal{D})| \cdot 2^{\cdot^{\cdot^{\cdot^{\scriptstyle 2^{\mathcal{O}(q_v(q + \log \omega))}}}}}
  \end{array}
  \right\rbrace
  \text{\scriptsize $q + 1.$}%
\end{equation*}

The correctness of the above approach is implied by \cref{combining-cor}.
\end{proof}

  \sv{ \section{Dynamic Programming Core for MSO$_1$-expressible Vertex Problems}\label{app:core}}
  \lv{ \subsection{Dynamic Programming Core for MSO$_1$-expressible Vertex Problems}\label{sec:MSOMonotone}}

\begin{restatable}{theorem}{thmMSOcore}\label{MSO-core}
  Let $\mathcal{P}$ be a vertex problem expressible in $MSO_1$ logic
  and $\Phi$ be the formula for this problem with $q$ quantified variables;
  $q_v$ of them being individual variables.
  There exists a monotone dynamic programming core $\mathfrak{C}$
  that (together with a proper $\mathfrak{C}$-vertex-membership function $\rho$)
  solves $\mathcal{P}$ such that
  \begin{align*}
      \left.\kern-\nulldelimiterspace
      \begin{array}{@{}c@{}}
        \tau(\mathfrak{C}, \mathcal{D}) = |V(\mathcal{D})| \cdot 2^{\cdot^{\cdot^{\cdot^{\scriptstyle 2^{\mathcal{O}(q_v(q + \log \omega))}}}}}
      \end{array}
      \right\rbrace
      \text{\scriptsize $q + 1$,} \\
      \left.\kern-\nulldelimiterspace
      \begin{array}{@{}c@{}}
        Size(\mathfrak{C}, \mathcal{D}) = 2^{\cdot^{\cdot^{\cdot^{\scriptstyle 2^{\mathcal{O}(q_v(q + \log \omega))}}}}}
      \end{array}
      \right\rbrace
      \text{\scriptsize $q$,}          
  \end{align*}
where $\omega$ is the width of $\mathcal{D}$.
\end{restatable}

\begin{proof}
  Let $\mathcal{P}$ be a vertex problem expressible in $MSO_1$ logic and $\Phi = \exists_{z_1} \Phi'(z_1)$ be an $MSO_1$ formula for this problem in the prenex normal form,
  where the value of $z_1$ represents the set $S$ such that $(G, S) \in \mathcal{P}$.
  A reduced partial evaluation tree $F_t^-$ for $\Phi$ and $t \in V(\mathcal{D})$ can be viewed as a set $\mathcal{F}_t^-$ of non-isomorphic trees of height $q - 1$ (subtrees of nodes of level 1 of $F_t^-$)
  with roots labeled by the choices of the value of $z_1$ and
  \begin{equation*}
  \left.\kern-\nulldelimiterspace
  \begin{array}{@{}c@{}}
    |\mathcal{F}_t^-| \leq 2^{\cdot^{\cdot^{\cdot^{\scriptstyle 2^{\mathcal{O}(q_v(q + \log \omega))}}}}}
  \end{array}
  \right\rbrace
  \text{\scriptsize $q$.}    
  \end{equation*}
  \begin{claim}
    Each tree $f^- \in \mathcal{F}_t^-$ is obtained either from a tree $f^{'-} \in \mathcal{F}_{t'}^-$, if $t = recolor(t', a \rightarrow b)$ or $t = addEdges(t', a \rightarrow b)$ by reduction or
    from two trees $f_1^- \in \mathcal{F}_{t_1}^-, f_2^- \in \mathcal{F}_{t_2}^-$ by combining them and then reducing, if $t = t_1 \oplus t_2$.    
  \end{claim}
  \begin{claim} \label{complete}
    If $t = recolor(t', a \rightarrow b)$ or $t = addEdges(t', a \rightarrow b)$, for each $f^{'-} \in \mathcal{F}_{t'}^-$
    there exists exactly one $f^- \in \mathcal{F}_t^-$ which is equal to reduced $f^{'-}$ and if $t = t_1 \oplus t_2$, for each $f_1^- \in \mathcal{F}_{t_1}^-, f_2^- \in \mathcal{F}_{t_2}^-$ there exists
    exactly one $f^- \in \mathcal{F}_t^-$ which is equal to reduced product of $f_1^-$ and $f_2^-$.    
  \end{claim}
  We define $Process_{\mathfrak{C}, \mathcal{D}}$ as follows.
  \begin{itemize}
    \item If $t = intro(v, a)$,
    \begin{equation*}
      Process_{\mathfrak{C}, \mathcal{D}}(t) = \mathcal{F}_t^-.
    \end{equation*}
  \item If $t = t_1 \oplus t_2$,
  \begin{align*}
    Process_{\mathfrak{C}, \mathcal{D}}(t) = \{\sv{&}(f^-, f_1^-, f_2^-)
    \in \mathcal{F}_{t}^- \times \mathcal{F}_{t_1}^- \times \mathcal{F}_{t_2}^-
    :
    f^- = reduce(f_1^- \otimes f_2^-))\},
  \end{align*}
  \item If $t = recolor(t', a \rightarrow b)$,
  \begin{align*}
    Process_{\mathfrak{C}, \mathcal{D}}(t) = \{\sv{&}(f^-, f^{'-})
    \in \mathcal{F}_{t}^- \times \mathcal{F}_{t'}^-
    :
    \sv{\\&} f^- = reduce(f^{'-})\}.
  \end{align*}
  \item If $t = addEdges(t', a, b)$,
  \begin{align*}
    Process_{\mathfrak{C}, \mathcal{D}}(t) = \{\sv{&}(f^-, f^{'-})
    \in \mathcal{F}_{t}^- \times \mathcal{F}_{t'}^-
    :
    \sv{\\&}f^- = reduce(f^{'-})\}.
  \end{align*}
  \end{itemize}
Observe that for each $t$,
\begin{align*}
  |Process_{\mathfrak{C}, \mathcal{D}}(t)| \sv{&}\leq \prod_{i = 1}^{\delta(t)} |\mathcal{F}_{t_i}^-| \leq
  \left.\kern-\nulldelimiterspace
  \begin{array}{@{}c@{}}
    2^{\cdot^{\cdot^{\cdot^{\scriptstyle 2^{\mathcal{O}(q_v(q + \log \omega))}}}}}
  \end{array}
  \right\rbrace
  \text{\scriptsize $q$.}
\end{align*}

We can compute $Process_{\mathfrak{C}, \mathcal{D}}(t)$ while computing $F_t^-$
in time at most
\begin{equation*}
  \left.\kern-\nulldelimiterspace
  \begin{array}{@{}c@{}}
    2^{\cdot^{\cdot^{\cdot^{\scriptstyle 2^{\mathcal{O}(q_v(q + \log \omega))}}}}}
  \end{array}
  \right\rbrace
  \text{\scriptsize $q + 1$.}
\end{equation*}
So the total time necessary to compute $Process_{\mathfrak{C}, \mathcal{D}}$ is at most
\begin{equation*}
  \left.\kern-\nulldelimiterspace
  \begin{array}{@{}c@{}}
    |V(\mathcal{D})| \cdot 2^{\cdot^{\cdot^{\cdot^{\scriptstyle 2^{\mathcal{O}(q_v(q + \log \omega))}}}}}
  \end{array}
  \right\rbrace
  \text{\scriptsize $q + 1$}.
\end{equation*}

Let $Accept_{\mathfrak{C}, \mathcal{D}}$ be the set of such $f^- \in \mathcal{F}_t^-$, where $t$ is the root of $\mathcal{D}$, that
$f^-$ with pruned external choices evaluates to true. Observe that $f^-$ with pruned external choices is a reduced evaluation tree
for $\Phi'(z_1)$ of height $q$, so its size (then also the time needed for evaluation) is bounded by
\begin{equation*}
  \left.\kern-\nulldelimiterspace
  \begin{array}{@{}c@{}}
    2^{\cdot^{\cdot^{\cdot^{\scriptstyle 2^{\mathcal{O}(q_v(q + \log \omega))}}}}}
  \end{array}
  \right\rbrace
  \text{\scriptsize $q$.}
\end{equation*}
So the time necessary to compute $Accept_{\mathfrak{C}, \mathcal{D}}$ is at most
\begin{align*}
  \left.\kern-\nulldelimiterspace
  \begin{array}{@{}c@{}}
    |\mathcal{F}_t^-| \cdot 2^{\cdot^{\cdot^{\cdot^{\scriptstyle 2^{\mathcal{O}(q_v(q + \log \omega))}}}}}
  \end{array}
  \right\rbrace
  \text{\scriptsize $q$} \sv{&}\leq
  \left.\kern-\nulldelimiterspace
  \begin{array}{@{}c@{}}
    2^{\cdot^{\cdot^{\cdot^{\scriptstyle 2^{\mathcal{O}(q_v(q + \log \omega))}}}}}
  \end{array}
  \right\rbrace
  \text{\scriptsize $q$}.
\end{align*}

It remains to define a proper membership function.
For $f^- \in \mathcal{F}_t^-$ let
\begin{equation*}
  \rho(f^-) =
  \begin{cases}
    0 & \text{if the root of } f^- \text{ is labeled by } \emptyset, \\
    1 & \text{otherwise}.
  \end{cases}
\end{equation*}

\begin{claim}\label{cl:DPcore}
  Dynamic programming core $\mathfrak{C}$ together with the membership function $\rho$ defined as above solves~$\mathcal{P}$.
\end{claim}
\begin{proof}[Proof of \cref{cl:DPcore}]
  Firstly, we prove that if $(G, S) \in \mathcal{P}$, then there exists a $(\mathfrak{C}, \mathcal{D})$-witness $\alpha$
  such that $S = S_\rho(\mathcal{D}, \alpha)$.
  Assume that $(G, S) \in \mathcal{P}$. Then there exists an unreduced partial evaluation tree $f$ for $\Phi'(z_1)$
  and $G$ such that the root of $f$ is labeled by $S$ and $f$ with pruned external choices evaluates to true.
  Then for $t_1, \ldots, t_n$ being the leaves of $\mathcal{D}$ there exist partial evaluation trees $f_1, \ldots, f_n$
  for $\Phi'(z_1)$ and $G_{t_1}, \ldots, G_{t_n}$ from which we can build $f$ by the induction on the decomposition tree,
  using the tree product operation in join nodes.
  For $i \in [n]$, let $S_i$ be the label of the root of $f_i$. Observe that $S_i = \emptyset$ or $S_i = \{v\}$ where $t_i = Leaf_{\mathcal{D}}(v)$.
  By the definition of tree product, $S = \bigcup_{i \in [n]} S_i$.
  By the definition of $\mathcal{F}_t^-$, for each $f_i$ there exist $f_i^- \in \mathcal{F}_{t_i}^-$ which is isomorphic to it.
  Observe that two partial evaluation trees $f_i, f_i^{-}$ with roots labeled by $\emptyset, \{v\}$ respectively cannot
  be isomorphic, so the label of the root of $f_i^-$ is the same as the label of the root of $f_i$.
  By \cref{complete}, there exist exactly one $(\mathfrak{C}, \mathcal{D}, f^-)$-witness $\alpha$
  such that $\alpha(t_i) = f_i^-$ for each $i \in [n]$. Then $S_\rho(\mathcal{D}, \alpha) = \bigcup_{i \in [n]} S_i = S$.
  Moreover, by \cref{combining-lemma}, $f^-$ is isomorphic to $f$.
  Then $f^-$ with pruned external choices evaluates to true, so $\alpha$ is a $(\mathfrak{C}, \mathcal{D})$-witness.

  Secondly, we prove that if there exists a $(\mathfrak{C}, \mathcal{D})$-witness $\alpha$
  such that $S = S_\rho(\mathcal{D}, \alpha)$, then $(G, S) \in \mathcal{P}$.
  Let $f^- = \alpha(t)$ where $t$ is the root of $\mathcal{D}$. Tree $f^-$ with pruned external choices evaluates to true.
  Let $t_1, \ldots, t_n$ be the leaves of $\mathcal{D}$. For $i \in [n]$, let $f_i^- = \alpha(t_i)$ and $S_i$ be the label of the root of $f_i^-$.
  Observe that $S_i = \emptyset$ or $S_i = \{v\}$ where $t_i = Leaf_{\mathcal{D}}(v)$.
  Then $S_\rho(\mathcal{D}, \alpha) = \bigcup_{i \in [n]} S_i$.
  Let $f$ be an unreduced partial evaluation tree that can be built from $f_1^-, \ldots, f_n^-$ by induction on the decomposition tree.
  Observe that by the definition of tree product, the label of the root of $f$ is $S = \bigcup_{i \in [n]} S_i$.
  By \cref{combining-lemma}, $f$ is isomorphic to $f^-$. Then $f$ with pruned external choices also evaluates to true,
  so $(G, S) \in \mathcal{P}$. \renewcommand{\qedsymbol}{$\diamondsuit$}
\end{proof}
\cref{cl:DPcore} concludes the proof of \cref{MSO-core}.
\end{proof}

Combining \Cref{MSO-core} with \cref{theorem}, we get the following corollary, which is a more explicit version of \Cref{cor:mainCor} stated in the introduction.

\begin{corollary}\label{diverse-MSO}
  Let $\mathcal{P}_1,  \ldots , \mathcal{P}_r$ be vertex problems expressible in $MSO_1$,
  $\Phi_i$ be an $MSO_1$ formula for $\mathcal{P}_i$, and $Div$ be a Venn diversity measure.
  The problem $Div^d(\mathcal{P}_1,  \ldots , \mathcal{P}_r)$ on a graph $G$ with given cliquewidth decomposition $\mathcal{D}$ can be solved in time
  \begin{equation*}
    \mathcal{O}\left(|V(\mathcal{D})| \cdot \prod_{i = 1}^r c_i \right),
   \end{equation*}
  where $c_i$ is a constant depending only on $\Phi_i$ and the width of $\mathcal{D}$.
\end{corollary}

\section{Conclusions and Open Problems}\label{sec:conclusions}
Our main contribution is that we can convert a monotone DP into one that solves a diverse version of the problem. 
We think it would be interesting to study this direction past the vertex problems.
For example, can we extend the definition of monotonicity to edge problems?
However, there is a barrier preventing FPT algorithms for many edge problems, as MSO$_2$ model checking is most likely not even in XP on cliques~\cite{CMR}.
Still, a transformation of the original DP algorithm might exist.
Indeed, this is in line with our approach, which also works independently of the running time of the original algorithm (of course, the final running time is dependent on it).
Take, for example, the Hamiltonian cycle problem.
There is an XP dynamic programming algorithm parameterized by the cliquewidth~\cite{BKK19}.
Is there an XP algorithm solving the diverse variant of the Hamiltonian cycle problem parameterized by the cliquewidth?

Another intriguing direction would be to research whether $d^{r^2}$ overhead in running time comparison of \Cref{mainThm} and \Cref{minmainThm} is necessary.

We believe that studying the Venn $f$-diversity measures is of independent interest.
It would be nice to see some systematic study of differences in such measures or perhaps some experimental study that would give us more understanding of the various measures.
Likewise, there has been a study comparing different diversity measures from the perspective of evolutionary computation~\cite{WO03}.

At this point, let us mention that for our purpose, we can allow even a more general definition of Venn $f$-diversity measure that also has access to the original object.
More precisely, we can use properties and shape of the input graph (as long as we only refer to the final graph $G$ and not to $G_t$ for $t$ not being a root of the decomposition).
For example, one can design a measure that penalizes/benefits vertices differently based on their degree in $G$.
Note that such a definition still allows us to get the same outcome for \Cref{mainThm}.
However, as we currently do not have a good application for this strengthening, we decided to go for a less general definition of the measures in the main body for the sake of simplicity and readability.

\lv{
 \bibliographystyle{alphaurl}
}
\aaai{
}
\ijcai{
 \bibliographystyle{named}
}
\bibliography{lit.bib}

\ijcai{
  \clearpage 
  \section*{Reproducibility Checklist}

This paper:

Includes a conceptual outline and/or pseudocode description of AI methods introduced \textbf{yes} %

Clearly delineates statements that are opinions, hypothesis, and speculation from objective facts and results  \textbf{yes} %

Provides well marked pedagogical references for less-familiare readers to gain background necessary to replicate the paper \textbf{yes} %

    \bigskip

Does this paper make theoretical contributions? \textbf{yes} %

If yes, please complete the list below.

    All assumptions and restrictions are stated clearly and formally. \textbf{yes} %

    All novel claims are stated formally (e.g., in theorem statements). \textbf{yes} %

    Proofs of all novel claims are included. \textbf{yes} %

    Proof sketches or intuitions are given for complex and/or novel results. \textbf{yes} %

    Appropriate citations to theoretical tools used are given. \textbf{yes} %

    All theoretical claims are demonstrated empirically to hold. \textbf{NA} %

    All experimental code used to eliminate or disprove claims is included. \textbf{NA} %

    \bigskip

    Does this paper rely on one or more datasets? \textbf{no} %

If yes, please complete the list below.
    A motivation is given for why the experiments are conducted on the selected datasets (yes/partial/no/NA)
    All novel datasets introduced in this paper are included in a data appendix. (yes/partial/no/NA)
    All novel datasets introduced in this paper will be made publicly available upon publication of the paper with a license that allows free usage for research purposes. (yes/partial/no/NA)
    All datasets drawn from the existing literature (potentially including authors’ own previously published work) are accompanied by appropriate citations. (yes/no/NA)
    All datasets drawn from the existing literature (potentially including authors’ own previously published work) are publicly available. (yes/partial/no/NA)
    All datasets that are not publicly available are described in detail, with explanation why publicly available alternatives are not scientifically satisficing. (yes/partial/no/NA)

    Does this paper include computational experiments? \textbf{no} %

If yes, please complete the list below.
This paper states the number and range of values tried per (hyper-) parameter during development of the paper, along with the criterion used for selecting the final parameter setting. (yes/partial/no/NA)
Any code required for pre-processing data is included in the appendix. (yes/partial/no).
All source code required for conducting and analyzing the experiments is included in a code appendix. (yes/partial/no)
All source code required for conducting and analyzing the experiments will be made publicly available upon publication of the paper with a license that allows free usage for research purposes. (yes/partial/no)
All source code implementing new methods have comments detailing the implementation, with references to the paper where each step comes from (yes/partial/no)
If an algorithm depends on randomness, then the method used for setting seeds is described in a way sufficient to allow replication of results. (yes/partial/no/NA)
This paper specifies the computing infrastructure used for running experiments (hardware and software), including GPU/CPU models; amount of memory; operating system; names and versions of relevant software libraries and frameworks. (yes/partial/no)
This paper formally describes evaluation metrics used and explains the motivation for choosing these metrics. (yes/partial/no)
This paper states the number of algorithm runs used to compute each reported result. (yes/no)
Analysis of experiments goes beyond single-dimensional summaries of performance (e.g., average; median) to include measures of variation, confidence, or other distributional information. (yes/no)
The significance of any improvement or decrease in performance is judged using appropriate statistical tests (e.g., Wilcoxon signed-rank). (yes/partial/no)
This paper lists all final (hyper-)parameters used for each model/algorithm in the paper’s experiments. (yes/partial/no/NA)
}

\end{document}